\documentclass[twoside,11pt]{article}

\usepackage{blindtext}

\usepackage{jmlr2e}

\usepackage[utf8]{inputenc}
\usepackage{enumitem}  
\usepackage{bm,amsfonts,amsmath, mathtools}
\usepackage{comment}
\usepackage{hyperref}
\usepackage{bbm}
\usepackage{subfig}

\usepackage{natbib}
\usepackage{color}
\usepackage{algorithm2e}

\newtheorem{thm}{Theorem}
\newtheorem{Assumption}{Assumption}
\newtheorem{Lemma}{Lemma}
\newtheorem{Corollary}{Corollary}

\usepackage{amsmath, amssymb}

\definecolor{forestgreen}{rgb}{0.0, 0.27, 0.13}

\newcommand{\x}{\mathbf{x}}
\newcommand{\y}{\mathbf{y}}
\newcommand{\Y}{\mathbf{Y}}
\newcommand{\X}{\mathbf{X}}

\renewcommand{\P}{\mathbb{P}}
\newcommand{\I}{\mathbb{I}}
\newcommand{\E}{\mathbb{E}}
\newcommand{\V}{\mathbb{V}}
\newcommand{\R}{\mathbb{R}}

\newcommand{\ourmethod}{\texttt{TRUST}}
\newcommand{\ourmethodpp}{\texttt{TRUST++}}

\usepackage{lastpage}

\firstpageno{1}
\ShortHeadings{Cabezas, Soares, Ramos, Stern and Izbicki}{Conformal Calibration of Statistical Confidence Sets}

\begin{document}

\title{Conformal Calibration of Statistical Confidence Sets}

\author{\name Luben M. C. Cabezas \email lucruz45.cab@gmail.com \\
       \addr Department of Statistics and Institute of Mathematics and Computer Science\\
       Federal University of São Carlos and University of São Paulo\\
       São Carlos, SP 13565-905 and 13566-590, Brazil
       \AND
       \name Guilherme P. Soares \email gui.pedri.soares@gmail.com \\
       \addr Institute of Mathematics and Computer Science\\
       University of São Paulo\\
       São Carlos, SP 13566-590, Brazil
       \AND
       \name Thiago R. Ramos \email  thiagorr@ufscar.br\\
       \addr Department of Statistics\\
       Federal University of São Carlos\\
       São Carlos, SP 13565-905, Brazil
       \AND
       \name Rafael B. Stern \email  rbstern@gmail.com\\
       \addr Institute of Mathematics and Statistics\\
       University of São Paulo\\
       São Paulo, SP 05508-090, Brazil
       \AND
       \name Rafael Izbicki \email  rafaelizbicki@gmail.com\\
       \addr Department of Statistics\\
       Federal University of São Carlos\\
       São Carlos, SP 13565-905, Brazil
}

\maketitle

\begin{abstract}
Constructing valid confidence sets is a crucial task in statistical inference, yet traditional methods often face challenges when dealing with complex models or limited observed sample sizes. These challenges are frequently encountered in modern applications, such as Likelihood-Free Inference (LFI). In these settings, confidence sets may fail to maintain a confidence level close to the nominal value.
In this paper, we introduce two novel methods, \ourmethod\ and \ourmethodpp, for calibrating confidence sets to achieve distribution-free conditional coverage.  
 These methods rely entirely on simulated data from the statistical model to perform calibration.
Leveraging insights from conformal prediction techniques adapted to the statistical inference context, our methods ensure both finite-sample local coverage and asymptotic conditional coverage as the number of simulations increases, even if $n$ is small. They effectively handle nuisance parameters and provide computationally efficient uncertainty quantification for the estimated confidence sets.
This allows users to assess whether additional simulations are necessary for robust inference.
Through theoretical analysis and experiments on models with tractable and intractable likelihoods, we demonstrate that our methods outperform existing approaches, particularly in small-sample regimes. This work bridges the gap between conformal prediction and statistical inference, offering practical tools for constructing valid confidence sets in complex models.
\end{abstract}

\section{Introduction}

Confidence sets are fundamental tools in statistical inference, allowing researchers to constrain the value of a parameter \(\theta \in \Theta\) based on observed data \(\x \in \mathcal{X}\). A confidence set \(R(\x) \subset \Theta\) is considered valid from a frequentist perspective if it satisfies the condition
\begin{equation}
\label{eq:cond_coverage}
\P\left( \theta \in R(\X)|\theta\right) =  1-\alpha   \ \ \forall \theta \in \Theta, 
\end{equation} 
where \(\alpha \in (0,1)\) is a predefined significance level. This means that \(R\) must achieve the correct coverage regardless of the true value of \(\theta\).

Traditional methods for constructing confidence sets, such as those detailed in standard textbooks \citep{degroot2012probability, schervish2012theory, casella2024statistical}, are often
inadequate when applied to complex modern models. For instance, constructing confidence sets for mixture models is notably challenging because standard asymptotic results may not hold \citep{chen2009hypothesis, wichitchan2019hypothesis}. Additionally, many traditional methods rely on asymptotic distributions, making them unsuitable for problems with small sample sizes ($n$). 

The challenges become more difficult in Likelihood-Free Inference (LFI) scenarios, where the statistical model is implicitly defined by a complex simulator of $\X|\theta$ and the likelihood function is intractable \citep{izbickiLeeSchafer, lueckmann2019likelihood, izbicki2019abc, papamakarios2019likelihood, cranmer2020frontier}.
In such cases, test statistics must be estimated, and as a result, traditional methods for constructing confidence sets often perform poorly \citep{dalmasso2021likelihood}.

This work aims to calibrate confidence sets to achieve conditional coverage (Equation \(\ref{eq:cond_coverage}\)), even for challenging models. We focus on confidence sets defined as:
\begin{align}
\label{eq:confidence}
R(\X) := \left\{\theta \in \Theta \mid \tau(\X,\theta) \geq C_\theta \right\},   \end{align}
where the cutoff \(C_\theta\) is chosen to ensure conditional coverage, and \(\tau\) measures the plausibility that \(\X\) was generated from \(\theta\) (in the language of hypothesis tests, $\tau$ is a test statistic). We assume \(\tau\) is fixed (see Sections \ref{sec:review_confidence} and \ref{sec:review_lfi} for details on its selection) and provide tools for calibrating \(C_\theta\).

Our calibration process only assumes the availability of a simulated dataset of independent pairs, \(\{(\theta_1, \X_1), \ldots, (\theta_B, \X_B)\}\), where each \(\theta_b\) is drawn from some reference distribution \(r(\theta)\) and each \(\X_b\) is generated from the statistical model with parameters \(\theta_b\), i.e., \(\X_b \sim \X \mid  \theta_b\). We assume that $\tau(\X,\theta)|\theta$ is strictly continuous for every $\theta$.
Note that each \(\mathbf{X}_b\) represents the entire dataset that may be observed, rather than a single sample point.

\textbf{Novelty.}  
While several methods have been proposed to calibrate \( C_\theta \) for general problems (Section \ref{sec:related_work}),
 our method offers several new contributions:

\begin{enumerate}
    \item \textbf{Distribution-free guarantees with $B$-asymptotic conditional coverage.}  
    Our approach constructs confidence sets \( \widehat R_B(\X) := \{\theta \in \Theta \mid \tau(\X,\theta) \geq \widehat C_\theta \}    \) with finite-sample local coverage:    
    $$\P\left(\theta \in \widehat{R}_B(\X)|\theta \in A \right)=1-\alpha,$$ 
    where $A$ is a subset of $\Theta$ chosen so that 
    $$\P\left(\theta \in \widehat{R}_B(\X)|\theta \in A \right) \approx \P\left(\theta \in \widehat{R}_B(\X)|\theta  \right).$$ 
    Additionally, our method provides $B$-asymptotic conditional coverage:    
    $$\lim_{B \longrightarrow \infty} \P\left(\theta \in \widehat{R}_B(\X)|\theta \right)=1-\alpha.$$

    Thus, our method is robust to poor estimates of \(C_{\theta}\), offering distribution-free guarantees while maintaining \(B\)-asymptotic validity. 
    Notably, this coverage is \emph{not} asymptotic with respect to the size of the observed dataset \( \mathbf{x} \), \( n \), as is common with traditional asymptotic approximations. Instead, it relies solely on the number of simulations, \( B \), which can generally be increased given sufficient computational resources. Furthermore, the set \(A\)  is chosen to ensure that local coverage approximates conditional coverage.
    
    \item \textbf{Uncertainty quantification.}  
In most cases, except for rare instances where \( C_\theta \) can be computed directly (Section \ref{sec:related_work}), current methods provide only an estimate \( \widehat{C}_\theta \) of the true \( C_\theta \). As a result, the confidence intervals produced are only approximations of the true intervals which inherently contain errors. However, these methods do not offer a way to quantify the uncertainty surrounding these estimated intervals. In contrast, our method explicitly constructs confidence sets that provide reliable uncertainty quantification around the estimated confidence intervals. Furthermore, this uncertainty quantification is computationally efficient and does not rely on costly procedures such as bootstrapping. Moreover, it  allows users to assess whether additional simulations are necessary for robust inference.
    
    \item \textbf{Effective handling of nuisance parameters.}  
Most models involve nuisance parameters, which are not of direct interest but can complicate inference. Current methods often struggle with properly handling nuisance parameters, particularly when they are numerous, as they typically require expensive numerical maximization algorithms in high-dimensional spaces to ensure coverage \citep{dalmasso2021likelihood}. In contrast, our method significantly reduces the computational burden by transforming the high-dimensional maximization problem into one of finding the maximum over a small, finite set of points, allowing for much more efficient optimization.
\end{enumerate}

\begin{figure}[!htb]
    \centering
    \subfloat[\centering 95\% confidence intervals. Blue dotted lines represent the lower and upper bounds for the oracle confidence interval.]{{\includegraphics[width=0.8\textwidth]{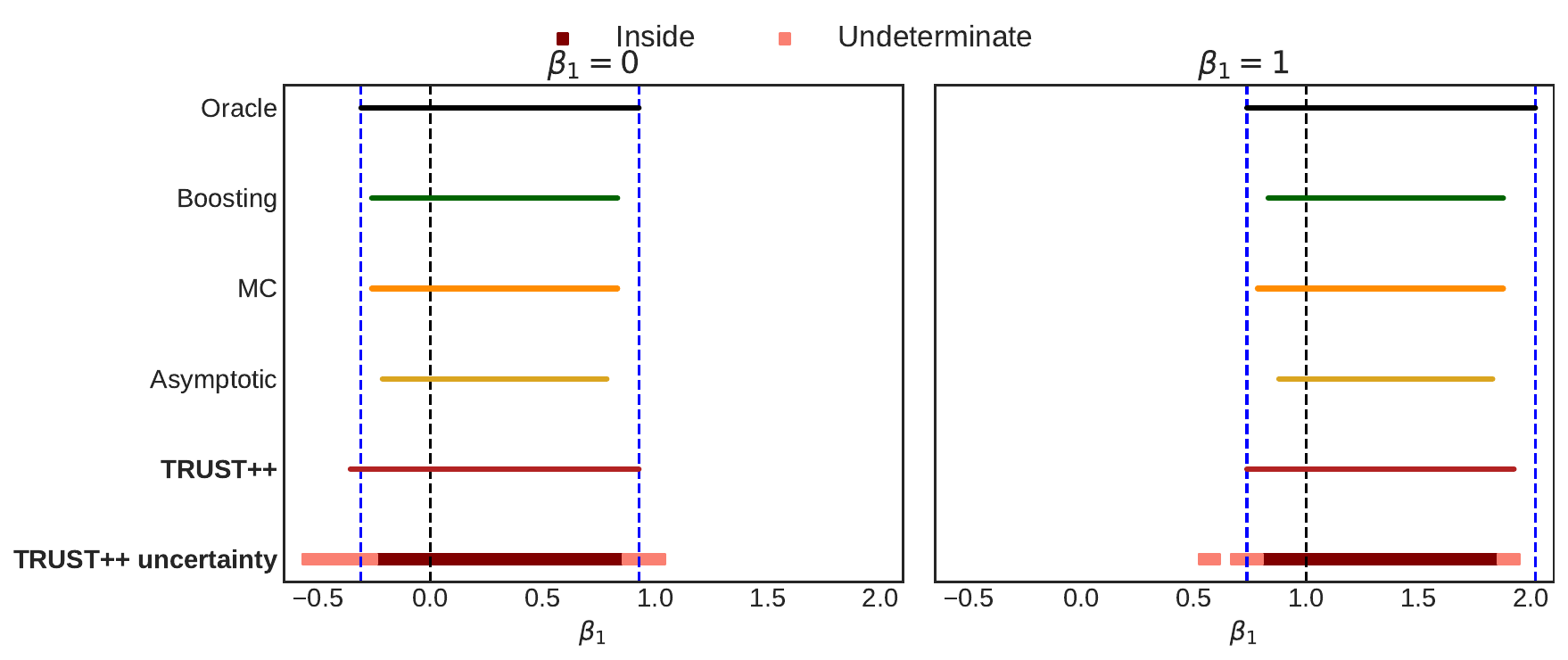}}}%
    \qquad
    \subfloat[\centering Deviance from oracle coverage.]{{\includegraphics[width=0.8\textwidth]{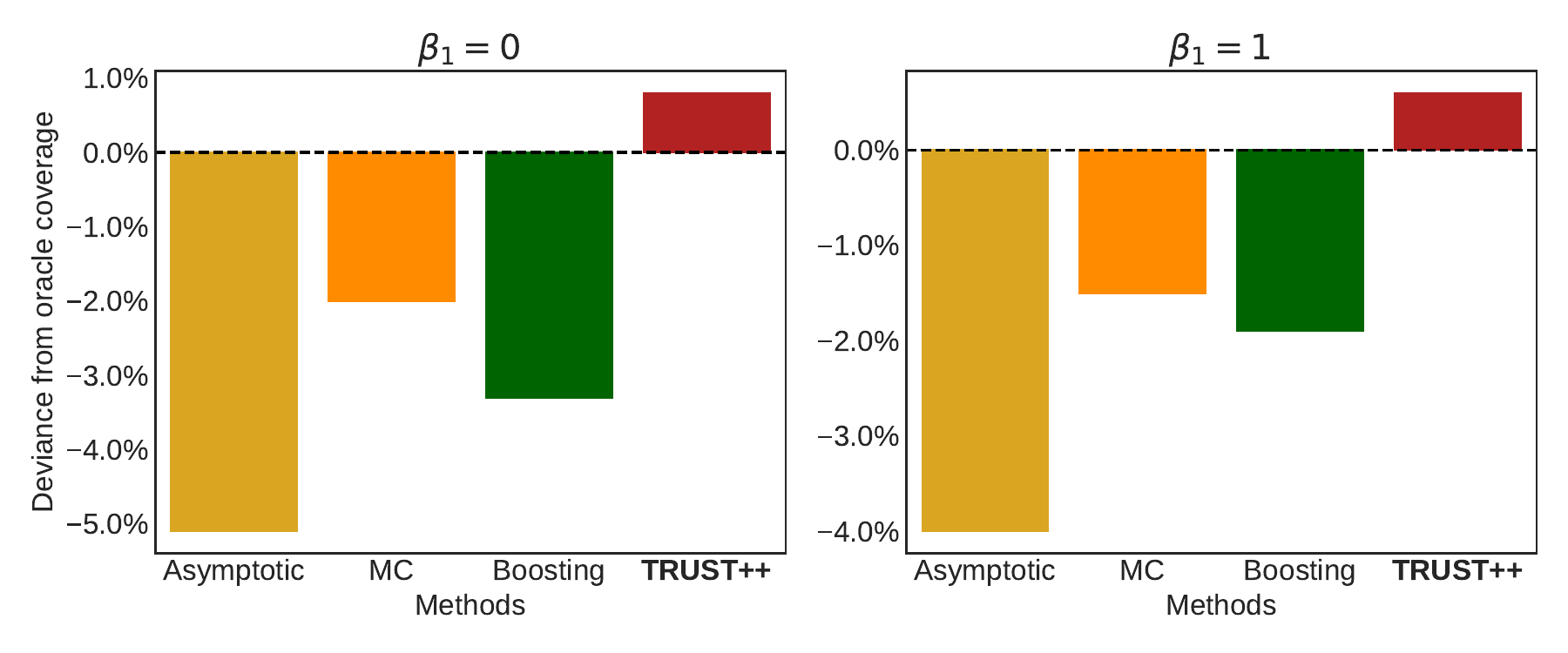}}}%
    \caption{Confidence intervals and deviance from oracle coverage comparisons for the GLM example with a sample size of $n = 50$. Our method closely approximates the oracle in terms of both the estimated confidence intervals (top) and coverage probabilities (bottom), while also providing uncertainty quantification for the intervals. This feature highlights regions where the true parameter value is more likely to lie. If the undetermined region is large, users can increase the number of simulations to reduce its size and refine the results. }%
    \label{fig:example_CI_and_prob_coverage_GLM}%
\end{figure}

Figure \ref{fig:example_CI_and_prob_coverage_GLM} compares \ourmethodpp\ with other methods for constructing confidence intervals for the slope coefficient in the generalized linear model explored in Section \ref{sec:nuisanceExamples}. Despite the presence of 3 nuisance parameters, \ourmethodpp\ produces intervals that are closer to the oracle confidence set, both in terms of size and empirical coverage. It even outperforms traditional GLM confidence sets based on $\chi^2$ approximations, as well as Monte Carlo-based intervals.

Our tools are developed using insights from conformal prediction techniques (see Section \ref{sec:related_work} for related work). However, while conformal methods focus on generating prediction sets for new labels $Y$, which are random variables, statistical inference aims to create confidence sets for a fixed parameter $\theta$ that generated the data $\x$. Thus, conformal techniques need to be adapted to fit this purpose. In particular, asymptotic conditional coverage, which is essential for valid sets, is not guaranteed by standard conformal methods.

We provide two methods:
\begin{itemize}
\item \ourmethod, which uses regression trees to offer a fast approximation to $C_\theta$.
\item \ourmethodpp, which enhances \ourmethod \ by employing bagging of trees to deliver more accurate estimates. The extension from trees to forests is non-trivial due to the need for finite sample guarantees.
\end{itemize}

The remaining of this work is organized as follows: Section \ref{sec:related_work} contains a discussion on related work. Section \ref{sec:methodology}   details the methodology, including the proposed methods and estimation procedures. Section \ref{sec:theory} presents the theoretical results. Section \ref{sec:implementation} discusses implementation details. Section \ref{sec:applications} provides applications to examples with tractable likelihood functions, likelihood-free inference, and nuisance parameter problems. Section \ref{sec:final_remarks} offers final remarks. Additional materials are provided in Appendices \ref{appendix:experiment_details} through \ref{appendix:proofs}, including experimental details, supplementary algorithms, and complete theoretical proofs. Finally, all code and supplementary materials are available in the GitHub repository:
 \href{https://github.com/Monoxido45/CSI}{https://github.com/Monoxido45/CSI}

\subsection{Relation to Other Work }
\label{sec:related_work}

\subsubsection{Classical Confidence Sets}
\label{sec:review_confidence}

Confidence sets typically follow the form described in Equation \ref{eq:confidence}. The function $\tau(\X,\theta)$ can be a pivotal quantity, which is a quantity whose distribution does not depend on $\theta$ (e.g., $\sigma^{-1}(X-\mu)$ when $X \sim N(\mu,\sigma)$ and $\sigma$ is known). Another example of $\tau(\X,\theta)$ is test statistics used for testing the null hypothesis that the data was generated from $\theta$, such as the likelihood ratio statistic $f(\x|\theta)/\sup_{\theta' \in \Theta} f(\x|\theta')$ \citep{lehmann1986testing}.

A key challenge is computing $C_\theta$. In simple cases, $C_\theta$ can be determined in closed form. For instance, when $\tau$ is a pivotal quantity, $C_\theta$ does not depend on $\theta$ and is typically straightforward to find \citep{casella2024statistical}. More often, asymptotic derivations are employed, which are effective only for large sample sizes $n$ and for certain test statistics \citep{van2000asymptotic}. However, these approximations require regularity conditions that frequently fail \citep{algeri2019searching}. Similarly, bootstrap  methods can be used, but these are  effective only for large sample sizes $n$ and also require regularity conditions \citep{van2000asymptotic}.

Our methods instead provide consistent (as $B$ increases) estimates of $C_\theta$  for any fixed $\tau$,
 even for small sample sizes $n$. We only require  a simulated dataset $\{(\theta_1,\X_1),\ldots,(\theta_B,\X_B)\}$.

\subsubsection{Likelihood-Free Inference}
\label{sec:review_lfi}

Likelihood-Free Inference (LFI) problems are problems in which  the statistical model lack a tractable analytical distribution for \(\mathbf{X}|\theta\), but it is still possible to generate simulations from it. For a comprehensive review, see \citet{cranmer2020frontier}. In such situations, $\tau$ needs to be defined from the simulated set as well. For instance, one may use $\tau(\x,\theta_0)$ as \(\widehat{\V}^{-1}[\theta|\x] (\widehat{\E}[\theta|\x]-\theta_0)^2\), where the conditional mean and variance are estimated using the simulated dataset \citep{masserano2023simulator}. For other possibilities, see e.g. \citet{brehmer2020mining,dalmasso2021likelihood}.

In LFI, once $\tau$ has been defined, $C_\theta$ is often estimated via Monte Carlo simulations for each fixed $\theta$, but this approach becomes prohibitively expensive in higher-dimensional spaces. Alternatively, \citet{dalmasso2020confidence, dalmasso2021likelihood, masserano2023simulator} observed that $C_\theta$ is the $1-\alpha$ conditional quantile of $\tau(\X,\theta)$ on $\theta$, and thus used quantile regressors to estimate $C_\theta$ using the simulated dataset $\{(\theta_1, \X_1), \ldots, (\theta_B, \X_B)\}$, where $\theta_b$ is drawn from a reference distribution \(r(\theta)\) and \(\X_b\) is drawn from the statistical model at $\theta_b$. These methods recover $C_\theta$ as \(B \to \infty\), but they do not offer finite-sample guarantees.


There is also a substantial body of work on Bayesian LFI, but these methods do not achieve conditional in the sense of Eq.~\ref{eq:cond_coverage}) even as \(B \to \infty\) \citep{hermans2021trust}.

\subsubsection{Conformal Prediction}
\label{sec:related_conformal}

Conformal methods have recently emerged as an approach to using data to obtain valid prediction regions under minimal assumptions \citep{vovk2005algorithmic,shafer2008tutorial,angelopoulos2023conformal}. Specifically, given exchangeable labeled data $\{(\X_1, Y_1), \ldots, (\X_m, Y_m)\}$, these methods create prediction regions \(R(\x)\) such that
\begin{align}
\label{eq:marginalcoverage}
\P\left(Y_{m+1} \in R(\X_{m+1}) \right) \geq 1 - \alpha.    
\end{align}
The first step in a standard conformal method involves creating a conformity score \(\hat{s}: \mathcal{X} \times \mathcal{Y} \to \mathbb{R}\), which measures how well the outputs of a regression model predict new labels \(Y \in \mathcal{Y}\). In a regression setting, a standard choice for the conformity score is the regression residual, given by \(\hat{s}(\x, y) = |y - \widehat{\E}[Y|\x]|\). The prediction region then takes the form \(R(\x) = \{y : \hat{s}(\x, y) \leq t\}\), where \(t\) is chosen to be the adjusted \( (1 - \alpha)\)-quantile of the residuals \(\hat{s}(\x, y)\) evaluated on a holdout set not used to train \(\hat{s}\).

The function \(\tau\) used to construct confidence sets can be interpreted as a conformal score \(\widehat{s}\), and $\theta$ can be interpreted as $Y$. However, statistical theory is typically employed to design \(\tau\) in a way that yields optimal confidence sets, leveraging the known distribution \(f(\mathbf{x} \mid \theta)\). Therefore, standard conformal scores are suboptimal for designing \(\widehat{s}\) for confidence sets. Additionally, most conformal scores   only deal with univariate \(Y\), whereas \(\theta\) is generally multivariate (see however \citet{dheur2024distribution} for approaches that can be used for multivariate $Y$'s).

Another key difference between prediction intervals and confidence sets is that \(\theta\) is not random, whereas \(Y\) is. Thus, marginal coverage (Eq.~\ref{eq:marginalcoverage}) is insufficient for statistical inference problems; typically, coverage for all fixed \(\theta\)'s, corresponding to conditional coverage (Eq.~\ref{eq:cond_coverage}), is also desired. This is because the distribution $r(\theta)$, used to sample $\theta$ on the training set, is often arbitrary; it does not encode any randomness observed in the real world.
Unfortunately, it is well-known that conditional coverage can't be achieved without strong assumptions \citep{lei2014distribution}. While many conformal approaches aim to achieve conditional coverage asymptotically, they only control coverage conditional on the features \(\x\), and not \(y\). The exception to this are label-conditional conformal prediction methods \citep{vovk2014conformal,vovk2016criteria,sadinle2019least,ding2023class}, but these apply only to classification problems.

Finally, we note that conformal sets have recently been applied to statistical inference \citep{patelvariational,baragatti2024approximate}. While these methods are highly useful, they are designed for a Bayesian framework, where the focus is on controlling marginal coverage over the prior distribution.

\subsubsection{Partition-Based Conformal Methods}

Our approach to asymptotically controlling conditional coverage is closely related to the works of \citet{vovk2005algorithmic,lei2014distribution,bostrom2020mondrian,izbicki2020flexible,izbicki2022cd,vovk2022algorithmic}. These methods first partition the feature space and then calibrate \( t \) using the conformal approach described in the previous section for each partition element separately. In particular, \ourmethod\ is closely related to LoCart \citep{cabezas2025regression}, which creates a data-driven partition of \(\mathcal{X}\) using regression trees designed to ensure that
\[
\P\left(Y_{n+1} \in R(\X_{n+1}) \mid \X_{n+1}\right) \approx 1 - \alpha.
\]
However, instead of partitioning the feature space \(\mathcal{X}\), we partition the parameter space \(\Theta\) to achieve coverage conditions on \(\theta\).

\subsubsection{Conformal Predictive Distributions}
\label{sec:predDist}

Conformal Predictive Distributions (CPDs) are a recent extension of the conformal prediction framework \citep{Vovk2019, Vovk2020,jonkers2024cpd}. These aim to model the entire predictive distribution of $Y$ in a non-parametric way. 
If a conformal $\hat{s}$ 
is isotonic and balanced \citep{Vovk2020}, the CPD for a new test instance $\X_{n+1}$ is defined as:
\begin{align}
\label{eq:CPD}
    Q((\X_{n + 1}, y), u) := \frac{1}{n} \sum_{i = 1}^n \I(\hat{s}(\X_i, Y_i) \leq \hat{s}(\X_{n + 1}, y)) + \frac{u}{n}  \; ,
\end{align}
for each $y \in \mathbb{R}$ and where $u \sim U(0,1)$ serves as a correction term to ensure the continuity of $Q((X_{n + 1}, y), u)$. 

In simple terms, $Q$ provides an estimated probability distribution for the test labels $y$ based on conformity scores. Specifically, if the data is i.i.d., $Q$ is valid, meaning that $Q((\X_{n + 1}, y), u)$ follows a uniform distribution \citep{Vovk2019}, and thus $Q$ is a valid cumulative distribution function. Leveraging the validity of $Q$, we construct $1-\alpha$ confidence CP prediction intervals by utilizing the specific $\alpha/2$ and $1 - \alpha/2$ percentiles of the CPD \citep{Bostroem2021}.

In this work, we use Conformal Predictive Distributions to estimate the distribution of the test statistic \(\tau(\mathbf{X}, \theta)\). This approach provides the desired distribution-free properties when estimating \(C_\theta\). Additionally, our CPDs allow for the computation of p-values for hypotheses of interest. However, directly using the definition from Eq.~\ref{eq:CPD} for statistical inference problems will not lead to consistent intervals for any fixed \(\theta\). The percentiles of \(Q\) would not depend on \(\theta\), which is the same issue discussed in Section~\ref{sec:related_conformal}. To address this, we employ a local CPD based on a partition of \(\Theta\). We note that \citet{Bostroem2021} also constructs CPDs based on partitions, but these partitions are made in \(\mathcal{X}\) rather than in \(\Theta\). While this approach is valuable for prediction tasks, it does not address the problem of statistical inference.

\section{Methodology}
\label{sec:methodology}

Fix a test statistic $\tau$, and
let  $H(\cdot|\theta)$ denote the 
distribution of   $\tau(\X,\theta)$ given $\theta$, that is, for each $\tau_0 \in \mathbb{R}$, let
$$H(\tau_0|\theta):=\P\left( \tau(\X,\theta)\leq \tau_0|\theta\right).$$

Now, notice that for the confidence set from Equation \ref{eq:confidence} to have a $(1-\alpha)$ confidence level, $C_\theta$ must satisfy
$$1-\alpha = \P(\theta \in R(\X) | \theta) = \P(\tau(\X, \theta) \geq C_\theta| \theta).$$
In other words, we need $\P(\tau(\X, \theta) \leq C_\theta) = \alpha$, which implies that $C_\theta$ is connected to the CDF $H$ by
$$C_\theta = H^{-1}(\alpha | \theta),$$
where $H^{-1}(\cdot | \theta)$ denotes the generalized inverse of $H(\cdot | \theta)$, given by
$$H^{-1}(\lambda) = \inf \{ t \in \R : H(t | \theta) \geq \lambda \}.$$
This suggests we derive confidence sets by estimating $H$ and then using the plugin cutoff  
\[
\widehat C_{\theta,B}=\widehat H^{-1}_{B}(\alpha|\theta),
\]  
where $\widehat H_B$ is the adjusted empirical distribution function defined as  
\[
\widehat H_B(t|\theta) := \frac{1}{B+1}\left(\sum_{b=1}^B \mathbf{1}\{\tau(\X^{(b)},\theta)\leq t\} + 1\right).
\]
The estimated confidence set will then be
$$ \widehat R_B(\X) := \left\{\theta \in \Theta \mid \tau(\X,\theta) \geq \widehat C_{\theta.B} \right\}.$$
The estimated conditional CDF $\widehat H_{B}$ can also be used to test statistical hypothesis of the type
$H_0:\theta=\theta_0$. Specifically, p-values
can be estimated via
$$\widehat  H_B(\tau(\x_{\text{obs}},\theta_0)|\theta_0),$$
where $\x_{\text{obs}}$ represents the observed data\footnote{Indeed, under the null hypothesis, 
$
  \P \left( H( \tau(\X_{\text{obs}},\theta_0))|\theta_0) \leq \alpha|\theta_0 \right)=\P \left(\tau(\X_{\text{obs}},\theta_0) \leq H^{-1}\left( \alpha|\theta_0\right) |\theta_0 \right) 
  = H\left(H^{-1}(\alpha) \right) 
  =\alpha,
$
and therefore 
$\widehat  H_B(\tau(\x_{\text{obs}},\theta_0)|\theta_0)$ will be a valid p-value if $H$ is consistently estimated as $B \longrightarrow \infty$.
}.

We build on this by introducing a novel approach to estimate $H(\cdot|\theta)$, which ensures robust properties.

\subsection{Estimating $H(\cdot|\theta)$}

Our first estimator of $H(\cdot | \theta)$, \ourmethod \ (Tree-based Regression for Universal Statistical Testing),   is constructed by partitioning $\Theta$. First, we show how any partition of $\Theta$ can be used to form an estimator of $H$. Following that, we discuss how to optimally select such a partition.

\subsubsection{Partition-based Estimate}
\label{sec:partitioned_based_estimate}

Let $\mathcal{A}$ be any fixed partition of $\Theta$
and \(\{(\theta_1, \X_1), \ldots, (\theta_B, \X_B)\}\) be the simulated dataset from the statistical model.
The cumulative distribution function $H(\cdot|\theta)$ can be estimated using the empirical cumulative distribution of the $\theta_b$ values from the simulated set that fall into the same partition element as $\theta$. Formally, for each $t \in \mathbb{R}$, we define
\begin{align}
\label{eq:empirical}
    \widehat{H}_B(t|\theta, \mathcal{A}) := \frac{1}{|I_{A(\theta)}|+1}\left(\sum_{b \in I_{A(\theta)}}\I\left(\tau(\X_b, \theta_b) \leq t\right)+1\right),
\end{align}
where $A(\theta)$ represents the element of $\mathcal{A}$ containing $\theta$, and $I_{A(\theta)} = \{b \in \{1,\ldots,B\} : \theta_b \in A(\theta)\}$ is the set of indices for all $\theta_b$ values that fall into $A(\theta)$. 

With this partition-based construction of $\widehat{H}$, the plugin cutoff $\widehat{C}_{\theta,B} = \widehat{H}^{-1}_B(\alpha|\theta)$, used to estimate confidence sets, corresponds to the adjusted $\alpha$-quantile of the values $\{\tau(\X_b, \theta_b) : b \in I_{A(\theta)}\}$. 
 Also, \(\widehat{H}_B\) is essentially a conditional predictive conformal distribution (Section \ref{sec:predDist}), with the exception that we do not add the uniform distribution used to ensure continuity as it is not needed here.
This method therefore aligns with a partition-based conformal approach. 
However, unlike standard partition-based conformal methods, the partition $\mathcal{A}$ must be selected to approximate $H(\alpha|\theta)$ well. The next section discusses how to choose $\mathcal{A}$ to achieve this goal.

\subsubsection{Choosing the partition: \ourmethod}
\label{sec:choosingA}

\ourmethod \ uses the estimator in Equation \ref{eq:empirical} for a specific
 partition of $\Theta$, $\mathcal{A}$. This partition is chosen
 by a data-driven optimization process designed to yield  an accurate estimate of $H$. Specifically, \ourmethod\ 
seeks for the coarsest partition 
such that all $\theta$'s that fall into the same element   will share a similar conditional distribution $H(\cdot|\theta)$.  This
choice will guarantee that  
\begin{align*}
    \P\left( \theta \in \widehat R_B(\X)|\theta\right)  \approx \P\left( \theta \in \widehat R_B(\X)|\theta \in A \right) =1-\alpha.
\end{align*}

In practice, \ourmethod \ creates a regression tree that uses $\theta$ as input and outputs $\tau(\x,\theta)$
using the simulated data
$(\theta_1,\tau(\X_1,\theta_1)),\ldots,(\theta_B,\tau(\X_B,\theta_B)),$
which is derived from the original dataset. 
This tree effectively partitions $\Theta$, and this partition has the desired property above  \citep[Theorem 1]{meinshausen2006quantile}.
 
Although our theoretical results assume that the data used to construct the optimal partition is separate from the data used to build the empirical CDF in Equation \ref{eq:empirical} (refer to Section \ref{sec:theory} for details), in practice, we do not perform this data split, as it did not yield any performance gains in our experiments. Algorithm \ref{alg:trust} outlines the complete \ourmethod\ procedure, and Figure \ref{fig:trust_flowchart} provides a graphical visualization of it.

\RestyleAlgo{ruled}
\SetKwComment{Comment}{/* }{ */}

\begin{algorithm}[hbt!]
\caption{\ourmethod \ algorithm}\label{alg:trust}
\KwData{Simulated dataset  $(\X_1, \theta_1), \dots, (\X_B, \theta_B)$; a grid $\Theta_{grid} \subseteq \Theta$}
\KwResult{The confidence set $ \widehat R_B(\X)$}
compute $\mathcal{I} = (\theta_1, \tau(\X_1, \theta_1)), \ldots, (\theta_B, \tau(\X_B, \theta_B))$\;
split $\mathcal{I}$ in $\mathcal{I}_{train}$ and $\mathcal{I}_{cal}$, where $\mathcal{I}_{train} \cap \mathcal{I}_{cal} = \emptyset$\;
fit a decision tree $\widehat g$ with $\mathcal{I}_{train}$: 
$\widehat g(\theta) \approx \mathbb{E}[\tau(\X, \theta) \mid \theta]$\;
$\widehat g$ will create a partition $\widehat{\mathcal{A}}_\text{train} = \{A_1, \dots, A_K\}$ of $\Theta$\;
\For{$\theta \in \Theta_{grid}$}{
  $A(\theta) \gets$ element of $\widehat{\mathcal{A}}_\text{train}$ where $\theta$ falls\; 
  $J_{A(\theta)} \gets \{i\in [\mathcal{I}_{cal}]: \theta_i \in A(\theta)\}$\;
  $\widehat{C}_{\theta,B} \gets$ adjusted $\alpha$-quantile of $\{\tau(X_c,\theta_c)\}_{c \in J_{A(\theta)}}$\;
}
\Return{$ \widehat R_B(\X) = \left\{\theta^* \in \Theta_{grid} \mid \tau(\X,\theta^*) \geq \widehat C_{\theta^*,B} \right\}$}
\end{algorithm}

\begin{figure}[!htb]
    \centering{{\includegraphics[width=0.8\textwidth]{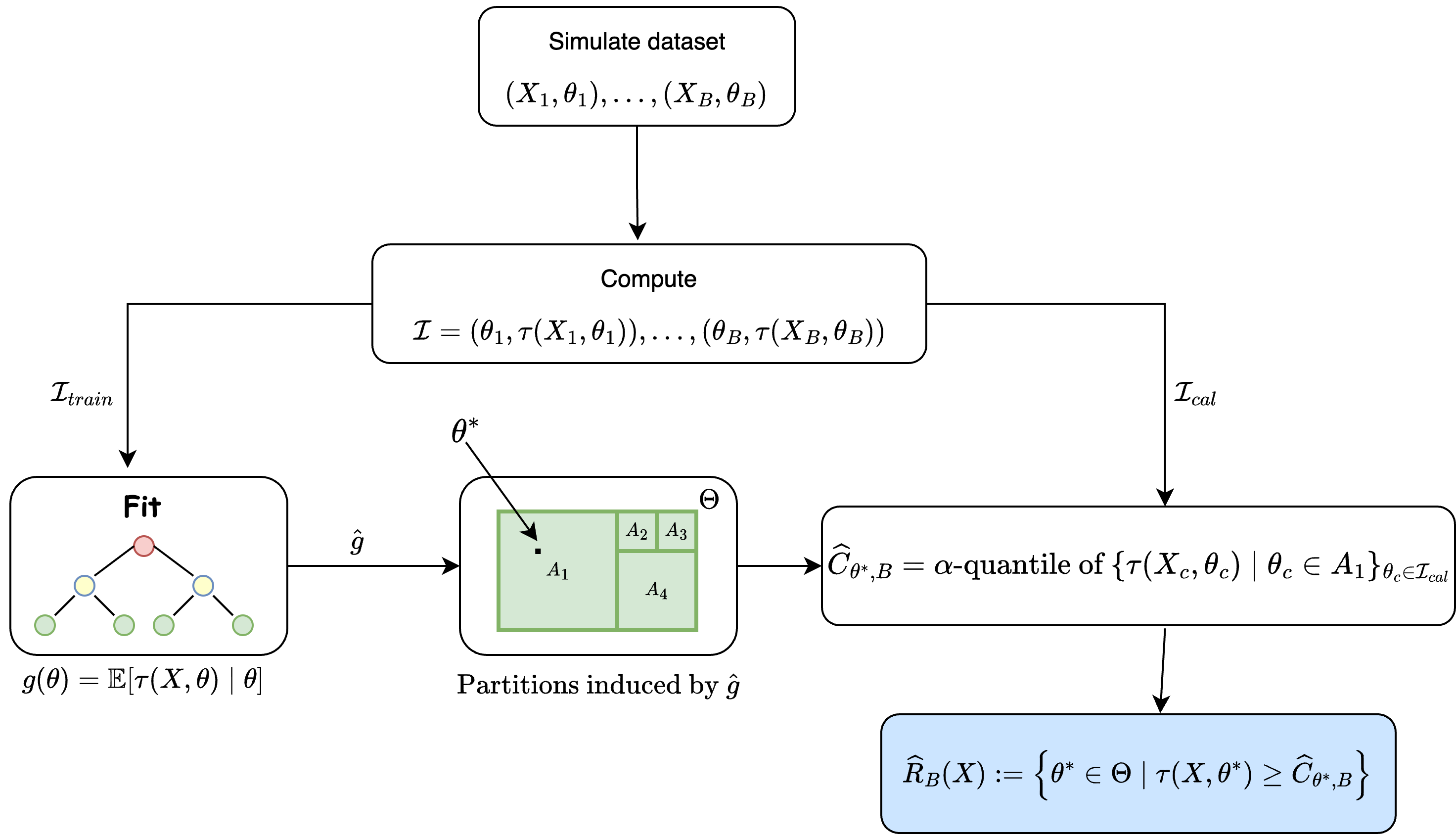}}}%
    \qquad
    \caption{Schematic flowchart of \ourmethod. Over a simulated dataset $(\X_1, \theta_1), \dots, (\X_B, \theta_B)$, we compute the statistic $\tau$ and create the set $\mathcal{I}$, which is splitted between $\mathcal{I}_{train}$ and $\mathcal{I}_{cal}$. $\mathcal{I}_{train}$ is used to train a regression tree that predicts $\tau(\X, \theta)$ using $\theta$. The tree partitions the parametric space in $A_1, \dots, A_k$. If a $\theta^*$, from the grid, falls inside the $A_1$ partition then, in order to calculate the cutoff $\widehat C_{\theta^*, B}$, we will compute the adjusted $\alpha$-quantile of all $\tau(\X, \theta)$'s from which $\theta$ belongs both to $\mathcal{I}_{cal}$ and $A_1$. Finally, we use the calculated cutoff to find the confidence set $\widehat R_{B}$.} %
    \label{fig:trust_flowchart}%
\end{figure}

\subsection{From trees to forests: \ourmethodpp}

Since a single regression tree often lacks expressiveness, we extend our approach to random forests. However,  averaging the \(\widehat{C}_{\theta, B}\) values obtained across trees does not ensure distribution-free guarantees \citep{cabezas2025regression}.
Instead, we use a random forest to design a new partition of $\Theta$, and then apply the methodology described in Section \ref{sec:partitioned_based_estimate}. Specifically, we first create $K$ \ourmethod \ regression trees, each on a different bootstrap sample of the simulated dataset $\{(\theta_1,\X_1),\ldots,(\theta_B,\X_B)\}$. Next, let $\rho(\theta',\theta)$ denote Breiman's proximity measure \citep{breiman2001random}, which counts the number of times $\theta'$ and $\theta$ appear in the same leaf across the $T$ trees. Finally, we define the partition $\mathcal{A}$ as the partition induced by the equivalence relation
$$\theta \sim \theta' \iff  \rho(\theta',\theta)=K.$$
That is, two $\theta$'s belong to the same element of $\mathcal{A}$ if and only if  they fall into the same leaves on all trees.
Thus, each element $A \in \mathcal A$ consists of parameter values that tend to co-occur in the same leaves across multiple trees in the random forest, and therefore share a similar value of $H(\cdot|\theta)$.

We also use a generalized version of \ourmethodpp \ that uses 
$$ A(\theta):=\{\theta' \in \Theta \mid \rho(\theta', \theta)\geq M \},$$
where $M$ is a tuning parameter 
in Equation \ref{eq:empirical}. $M$ can be interpreted as the minimal number of trees required to vote for $\theta$ and $\theta'$ to be in the same partition element. See Section \ref{sec:implementation} for details on how we choose it in practice.
Although this version does not always control coverage for $M<K$ \citep{guan2023localized}, since it does not induce a partition,  in practice it  has good performance. 

\ourmethodpp\ can also be framed as a specific instance of the LCP framework by utilizing a box kernel with the Breiman proximity measure  replacing the commonly used Euclidean distance.  This adaptation allows us to leverage the conformalization procedures from \cite{guan2023localized} and \cite{hore2023conformal} to transform \ourmethodpp\ into a proper conformal method, thereby ensuring exact coverage control. In this setting, the tuning parameter $M$ serves as a kernel bandwidth, regulating the granularity of the partition $\mathcal{A}$.

Algorithm \ref{alg:trustpp} summarizes the non-conformalized version of our method, which we name \ourmethodpp, and Figure \ref{fig:trustpp_flowchart} shows the flowchart of the procedure.

\begin{algorithm}[hbt!]
\caption{\ourmethodpp \ algorithm}\label{alg:trustpp}
\KwData{Simulated dataset  $(\X_1, \theta_1), \dots, (\X_B, \theta_B)$; a grid $\Theta_{grid} \subseteq \Theta$}
\KwResult{The confidence set $ \widehat R_B(\X)$}
compute $\mathcal{I} = (\theta_1, \tau(\X_1, \theta_1)), \ldots, (\theta_B, \tau(\X_B, \theta_B))$\;
split $\mathcal{I}$ in $\mathcal{I}_{train}$ and $\mathcal{I}_{cal}$, where $\mathcal{I}_{train} \cap \mathcal{I}_{cal} = \emptyset$\;
draw $K$ bootstrapped samples from $\mathcal{I}_{train}$\;
fit $K$ \ourmethod \ trees, each one with a different sample (Alg.~\ref{alg:trust})\;

\For{$\theta \in \Theta_{grid}$}{
$A(\theta) \gets \{\theta' \in \Theta \mid \rho(\theta', \theta) \geq M\}$, where $\rho(\theta', \theta)$ is the Breiman's proximity measure and M a tuning parameter\;
$J_{A(\theta)} \gets \{i\in [\mathcal{I}_{cal}]: \theta_i\in A(\theta)\}$\;
$\widehat{C}_{\theta,B} \gets$ adjusted $\alpha$-quantile of $\{\tau(\X_c,\theta_c)\}_{c \in J_{A(\theta)}}$\;
}

\Return{$ \widehat R_B(\X) = \left\{\theta^* \in \Theta_{grid} \mid \tau(\X,\theta^*) \geq \widehat C_{\theta^*,B} \right\}$}
\end{algorithm}

\begin{figure}[!htb]
    \centering{{\includegraphics[width=0.8\textwidth]{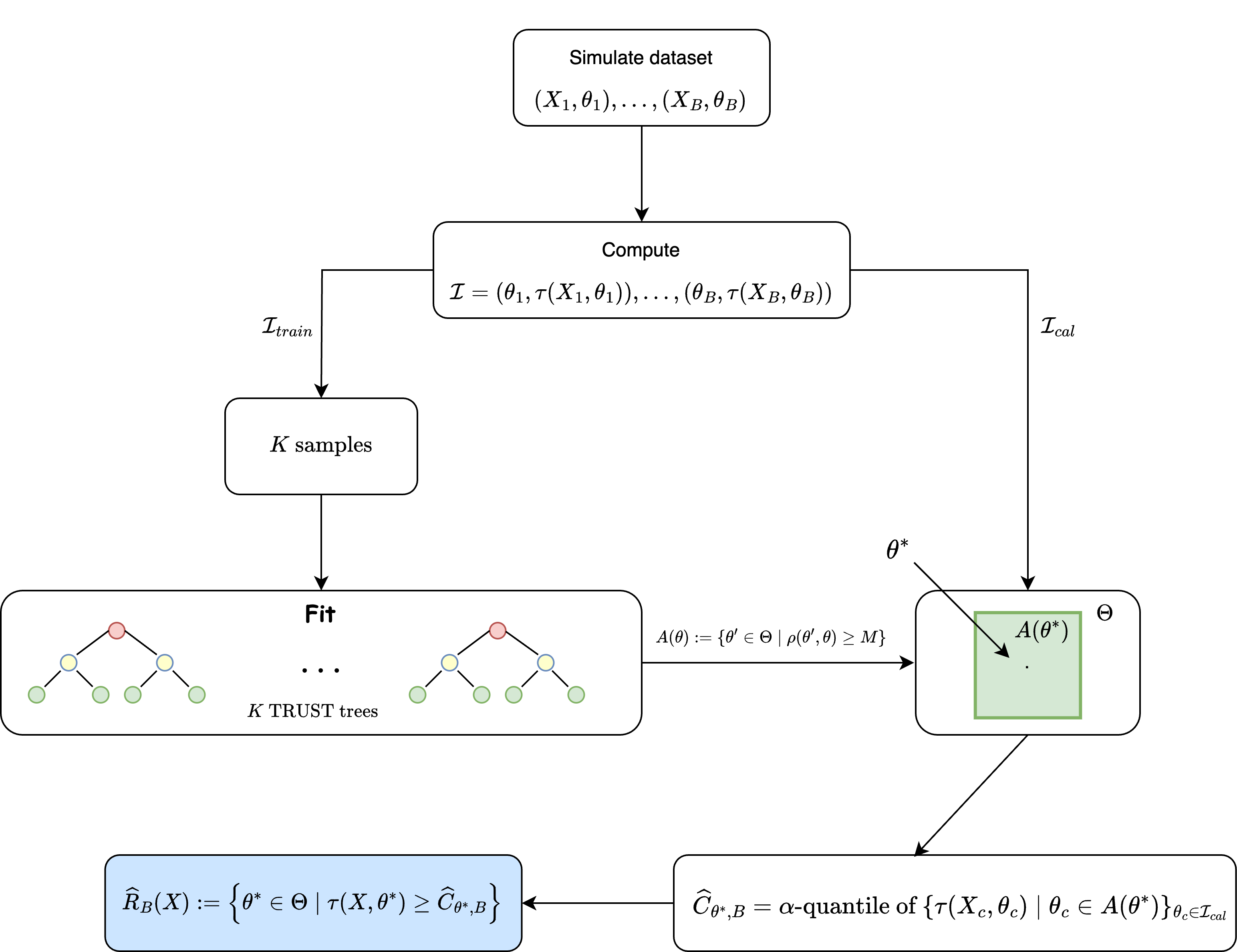}}}%
    \qquad
    \caption{Schematic flowchart of \ourmethodpp. Over a simulated dataset $(\X_1, \theta_1), \dots, (\X_B, \theta_B)$, we compute the statistic $\tau$ and create the set $\mathcal{I}$. Then, we draw K bootstrapped samples over the split $\mathcal{I}_{train}$ and use each sample to fit a different \ourmethod\ tree. The set $A(\theta)$ is calculated using Breiman's proximity measure, $\rho$, which defines whether two $\theta, \theta'$ belong to the same partition element. The cutoff $\widehat C_{\theta^*, B}$ is then calculated as the $\alpha-$quantile of all $\tau$ statistics of their respective $\theta$ that falls over $A(\theta)$, for the $\theta$'s that belong to $\mathcal{I}_{cal}$. Lastly, the estimated confidence set $\widehat R_{B}$ can be calculated using $\widehat C_{\theta^*, B}$.}%
    \label{fig:trustpp_flowchart}%
\end{figure}

\subsection{Nuisance parameters}

Statistical models often have nuisance parameters, which are parameters that are not of direct interest but must be accounted for in the analysis.
 Suppose that the vector of parameters can be decomposed \(\theta = (\mu, \nu)\), where \(\mu \in M\) are the parameters of interest and \(\nu \in N\) are nuisance parameters. Our goal is to construct a prediction region for \(\mu\) alone.

Let \(\tau(\X, \mu)\) be a metric that evaluates how plausible it is that \(\mu\) generated \(\X\). Confidence sets typically have the following form
\[ R(\x) := \left\{\mu : \tau(\x, \mu) \geq {C}_{\mu} \right\}. \]
To control conditional coverage for every $\mu$ and $\nu$, we must set 
\begin{align}
\label{eq:oracleNuisance}
    {C}_{\mu} = \inf_{\nu} {C}_{\mu, \nu},
\end{align}
where \({C}_{\mu, \nu}\) controls conditional coverage at \((\mu, \nu)\) \citep{dalmasso2021likelihood}. We now describe how to estimate ${C}_{\mu}$ using our proposed methods.

Let \(H(\cdot | \mu, \nu)\) be the distribution of \(\tau(\X, \mu)\) conditional on \((\mu, \nu)\). We use the same methodology, \ourmethod\ or \ourmethodpp, to estimate \(H\), namely by using trees to regress $\tau(\X, \mu)$ on \((\mu, \nu)\). Let \(\widehat{H}_B\) be this estimate. Then, we estimate \({C}_{\mu}\) via
\begin{align}
\label{eq:cNuisance}
\widehat{C}_{\mu, B} = \inf_{\nu} \widehat{H}_{B}^{-1}(1 - \alpha | \mu, \nu).    
\end{align}

Although computing the infimum in this equation might seem computationally expensive, the problem is simplified because \(\widehat{H}_B\) is computed using a   partition of \(\Theta\). This means that the minimization only requires evaluating \(\widehat{H}_{B}^{-1}(1 - \alpha | \mu, \nu)\) for a finite number of \(\nu\) values, one for each element of the partition  that  covers $\mu$. For example, in the case of \ourmethod, this approach only requires minimization over the leaves associated with the given \(\mu\).

In \ourmethodpp, we exploit the fact that only the values of $\nu$ that appear in the tree splits are relevant when computing the infimum in Equation \ref{eq:cNuisance}. Specifically, let $\mathbb{A} \subset N$ represent the set of nuisance parameter values involved in the splits of the trees generated by \ourmethodpp. Then, for each fixed $\mu$, we can express the infimum as
$$
\inf_{\nu} \widehat{H}_{B}^{-1}(1 - \alpha | \mu, \nu) = \min_{\nu \in \mathbb{G}} \widehat{H}_{B}^{-1}(1 - \alpha | \mu, \nu),
$$
where $\mathbb{G} = \{\nu \pm \epsilon : \nu \in \mathbb{A} \}$, and $\epsilon$ is a vector of small numbers that perturbs the nuisance parameter values slightly. This works because, for any value of \(\nu\), there exists a corresponding \(\nu\) in the grid \(\mathbb{G}\) such that both values fall into the same leaf of the tree. Consequently, it suffices to compute the minimum over \(\nu \in \mathbb{G}\). To ensure the $\epsilon$-neighborhoods around each $\nu$ value do not overlap, $\epsilon$ is chosen as one-third of the minimum coordinate-wise distance between all distinct $\nu \in \mathbb{A}$. 

To reduce the number of grid values, we may select $\mathbb{A}$ to include only the most relevant nuisance values—those associated with the splits closer to the root of the trees \citep{hastie2009elements}. This approach greatly decreases the computational burden of gridding $N$, which other methods require, and consequently limits their applicability to low-dimensional settings with few nuisance parameters.


\subsection{Uncertainty About the Confidence Sets}\label{sec:uncertainty_confidence_set}


Our confidence set \( \widehat R_B(\X) := \{\theta \in \Theta \mid \tau(\X,\theta) \geq \widehat C_\theta \}    \) is an estimate of the oracle set
\(  R(\X) := \{\theta \in \Theta \mid \tau(\X,\theta) \geq  C_\theta \}\), as is the estimate from the other approaches described in Section \ref{sec:related_work}.
 In this section, we explore another advantage of \ourmethod \ 
and \ourmethodpp:
  They offer a computationally efficient way to approximate the uncertainty associated with these sets, which we describe in what follows.

Recall that, for each $\theta$, $C_\theta$ is estimated by computing the adjusted
$\alpha$-quantile of the set
$T_{A(\theta)} = \{\tau(\X_b, \theta_b) : b \in I_{A(\theta)}\}$, while $C_\theta$ itself is the $\alpha$-quantile of the distribution of $\tau(\X, \theta) | \theta$. Let $\tau_{A(\theta)}^{(i)}$ be the $i$-th order statistics of $T_{A(\theta)}$ and define $Z := |\{ \delta \in T_{A(\theta)}: \delta \leq C_{\theta} \}|$, where $Z$ represents the number of statistics in $T_{A(\theta)}$ that are less than or equal to the cutoff $C_{\theta}$. Note that if the statistics in $T_{A(\theta)}$ are identically distributed, then $Z \sim \text{Binomial}(n, 1 - \beta)$. Using this result, we adapt the method by \cite{hahn2011statistical} for constructing confidence intervals for quantiles to create a $1 - \beta$ confidence set for $C_\theta$. This confidence set is defined as:
\begin{align}
\label{eq:cutoff_confidence_interval}
[\tau_{A(\theta)}^{(l)}, \tau_{A(\theta)}^{(u)}],
\end{align}
where \( l \) and \( u \) are selected to satisfy constraints ensuring both the validity and informativeness of the confidence interval. Specifically, \( u \) and \( l \) are selected such that \( l \leq u \leq |I_{A(\theta)}| \)  are as close as possible and satisfy the restriction
\begin{align}
\label{eq:restriction_CI}
    \P(l \leq Z \leq u - 1) \geq 1 - \beta \;,
\end{align}
yielding a valid confidence interval. This is because
\begin{align*}
\P(\tau_{A(\theta)}^{(l)} \leq C_\theta \leq \tau_{A(\theta)}^{(u)}) = \P(l \leq Z \leq u - 1) \geq 1 - \beta \; .
\end{align*}
Thus, by choosing $u$ and $l$ as narrowly as possible satisfying the restriction from Eq. \ref{eq:restriction_CI} we get the thinnest valid confidence interval for $C_{\theta}$. Additionally, despite the i.i.d assumption about statistics in $T_{A(\theta)}$ being not strictly true, the values in $T_{A(\theta)}$ are approximately identically distributed due to the way the partition was chosen (Section \ref{sec:choosingA}).


For each $\theta \in \Theta$, let $(\widehat{C}^L_\theta, \widehat{C}^U_\theta)$ represent the resulting confidence interval for $C_\theta$. The uncertainty can be propagated to the estimated sets by calculating the following three sets:
\begin{align*}
  &\mathcal{I}(\X)=\{\theta \in \Theta \mid \tau(\X,\theta) \geq  \widehat{C}^U_\theta \},  \\
  &\mathcal{O}(\X)=\{\theta \in \Theta \mid \tau(\X,\theta) \leq  \widehat{C}^L_\theta \}, \text{ and }\\
  &\mathcal{U}(\X)=\{\theta \in \Theta \mid \widehat{C}^L_\theta \leq \tau(\X,\theta) \leq  \widehat{C}^U_\theta \},
\end{align*}
where $\mathcal{I}(\X)$ represents the set of parameter values confidently inside the confidence interval, $\mathcal{O}(\X)$ contains values confidently outside, and $\mathcal{U}(\X)$ includes those where the status remains uncertain.
In the terminology of 3-way hypothesis testing \citep{berg2004no,esteves2016logical,Izbicki2025REACT},  $\mathcal{I}(\X)$  is the acceptance region,  $\mathcal{O}(\X)$ is the rejection region, and  $\mathcal{U}(\X)$ is the agnostic region.
The size of $\mathcal{U}(\X)$ can be decreased by increasing the number of simulations used to estimate $C_\theta$, $B$.
The next theorem shows that this approach controls the probability of incorrect conclusions.

\begin{thm}
\label{thm:uncertainty_coverage}
    Fix $\x \in \mathcal{X}$ and $\theta \in \Theta$. Let 
    $(\widehat{C}^L_\theta, \widehat{C}^U_\theta)$ represent a $(1-\beta)$-level confidence interval for $C_\theta$ such that $\P\left(C_\theta \leq \widehat{C}^L_\theta \right)=\P\left(\widehat{C}^U_\theta \leq C_\theta \right)$.
    Then 
    $$\P\left(\theta \in \mathcal{I}(\x) |\theta \notin R(\x) \right)\leq \beta/2$$
    and
    $$\P\left(\theta \in \mathcal{O}(\x)|\theta \in R(\x) \right)\leq \beta/2.$$
\end{thm}

This method can be readily extended to the setting with nuisance parameters. Let $\eta^*$ be the configuration of nuisance parameters for which $\widehat{C}_{\mu, B}$ is computed (Equation \ref{eq:cNuisance}). By construction, $\widehat{C}_{\mu, B}$ is the adjusted $\alpha$-quantile of the set
$T_{A(\mu, \eta^*)} = \{\tau(\X_b, \mu_b,\eta_b) : b \in I_{A(\mu,\eta^*)}\}$. Thus, we can derive a confidence set for the optimal cutoff using the same approach as in Eq. \ref{eq:cutoff_confidence_interval}, based now on the set $T_{A(\mu, \eta^*)}$.

\section{Theoretical Results}
\label{sec:theory}

The key aspect of our method lies in the fact that achieving a good approximation of \( H(\cdot|\theta) \) allows the partitions in \ourmethod{} and \ourmethodpp{} to effectively capture the local behavior of the test statistic \( \tau \).

Building on this, our theoretical framework relies on ensuring that \( \widehat{H}_B \) closely approximates \( H \) in both \ourmethod{} and \ourmethodpp{}. This idea is formalized in the following assumption:

\begin{Assumption}\label{assumption:strong_consistency}
Let \( \widehat{H}_B \) represent the approximation of \( H \) under either \ourmethod{} or \ourmethodpp{}. For any \( \varepsilon > 0 \) and \( \delta > 0 \), there exists a \( B_0 \in \mathbb{N} \) such that, for all \( B \geq B_0 \),
\begin{align*}
    \mathbb{P}\left(\sup_{t \in \mathbb{R}, \theta' \in \Theta} \left|\widehat{H}_B(t | \theta') - H(t | \theta')\right| \leq \varepsilon \right) \geq 1 - \delta.
\end{align*}
\end{Assumption}

This assumption ensures that, with high probability, \( \widehat{H}_B \) approximates \( H \) as \( B \) increases. This result is supported by theoretical findings, such as those of \cite{meinshausen2006quantile, consistency_gabor}, on the consistency of tree-based models. To guarantee this consistency, certain conditions regarding the distribution of covariates and the structure of the trees are necessary. In tree construction, it is important that node sizes are balanced, with the proportion of observations in each node decreasing as the total number of observations grows. Additionally, each variable must have a minimum probability of being selected for node splitting, and splits should ensure a balanced distribution of observations between subnodes. These assumptions are quite reasonable and enable the conditional distribution estimates to be consistent with the true distribution, providing a solid theoretical foundation for our method.

\subsection{Partition  Coverage Guarantees}

As discussed in Section \ref{sec:partitioned_based_estimate}, given any fixed partition \( \mathcal{A} \) of \( \Theta \), the cumulative distribution function \( H(\cdot|\theta) \) can be estimated empirically using the values of \( \theta_b \) that fall within the same partition element as \( \theta \). Although a data-agnostic partitioning approach does not guarantee that \( \widehat{H} \) will closely approximate \( H \), we can still ensure that the plugin cutoff \( \widehat{C}_{\theta,B} = \widehat{H}^{-1}_B(\alpha|\theta) \), corresponding to the \( \alpha \)-quantile of the values \( \{\tau(\X_b, \theta_b) : b \in I_{A(\theta)}\} \), achieves the desired coverage when conditioned on the partition element. This result is formalized in the following theorem.

\begin{thm}[Partition-Based Coverage]
\label{thm:partition-based-coverage}
Let \(\{(\theta_1, \X_1), \dots, (\theta_B, \X_B)\}\) be an i.i.d. simulated dataset with a fixed partition \(\mathcal{A}\) of \(\Theta\), where each \(\theta_b\) is drawn from a reference distribution \(r(\theta)\) and each \(\X_b\) is generated according to the statistical model with parameters \(\theta_b\). For a test statistic \(\tau(\X, \theta)\), consider the confidence set constructed by our method:
\[
\widehat{R}_B(\X) = \{\theta \in \Theta \mid \tau(\X, \theta) \geq \widehat{C}_{\theta, B}\},
\]
where \(\widehat{C}_{\theta, B} = \widehat{H}^{-1}_B(\alpha | \theta)\) represents the adjusted \(\alpha\)-quantile of the values \(\{\tau(\X_b, \theta_b) : b \in I_{A(\theta)}\}\), conditioned on \(\theta\) belonging to the partition element \(A(\theta)\). 
Then, we have
\[
\P(\theta \in \widehat{R}_B(\X) \mid \theta \in A(\theta)) \geq 1 - \alpha,
\]
ensuring the desired coverage probability of the confidence set, conditioned on the partition element.
 
\end{thm}

This result applies broadly to any partition, yet it carries particular significance for those formed using \ourmethod{} and \ourmethodpp{}. For these methods, we specifically anticipate the  relationship
\begin{align*}
1 - \alpha=\P\left( \theta \in \widehat{R}_B(\X) \mid \theta \in A(\theta) \right) \approx \P\left( \theta \in \widehat{R}_B(\X) \mid \theta \right),
\end{align*}
which suggests that, with these methods, the probability that $\theta$ lies within the estimated region $\widehat{R}_B(\X)$, given its association with $A(\theta)$, closely aligns with the nominal confidence level $(1 - \alpha)$. In the following section, we will state the theorem formally. Also, notice that although the local guarantee depends on the choice of $r(\theta)$ (and thus has a Bayesian flavor), the conditional one does not.

Beyond the conditional coverage established above, these results also imply marginal coverage:

\begin{Corollary}[Marginal Coverage]
\label{cor:partition-based-marginal-coverage}
Let \(\{(\theta_1, \X_1), \dots, (\theta_B, \X_B)\}\) be an i.i.d. simulated dataset with a fixed partition \(\mathcal{A}\) of \(\Theta\), where each \(\theta_b\) is drawn from a reference distribution \(r(\theta)\) and each \(\X_b\) is generated according to the statistical model with parameters \(\theta_b\). For a test statistic \(\tau(\X, \theta)\), consider the confidence set constructed by our methods. Then
$
\P(\theta \in \widehat{R}_B(\X)) \geq 1 - \alpha.
$
\end{Corollary}



\subsection{\ourmethod{} and \ourmethodpp{} Conditional  Coverage Guarantees}
\label{sec:coverage_guarantees}


Both \ourmethod{} and \ourmethodpp{} employ regression trees that satisfy the consistency Assumption \ref{assumption:strong_consistency}, as guaranteed by established  results for tree-based models \citep{meinshausen2006quantile, consistency_gabor}. With these partitions, we not only guarantee partition-based coverage, as indicated in Theorem \ref{thm:partition-based-coverage}, but also ensure that both methods asymptotically achieve optimal conditional coverage, as formalized in the next theorem.

\begin{thm}[$B$-Asymptotic Conditional Coverage]\label{thm:asympt_conditional}
Let \(\{(\theta_1, \X_1), \dots, (\theta_B, \X_B)\}\) be an i.i.d. simulated dataset, where each \(\theta_b\) is drawn from a reference distribution \(r(\theta) > 0\) and each \(\X_b\) is generated according to the statistical model with parameters \(\theta_b\). For a test statistic \(\tau(\X, \theta)\), consider the confidence set constructed by either \ourmethod{} or \ourmethodpp{}:
\[
\widehat{R}_B(\X) \coloneq \{\theta \in \Theta \mid \tau(\X, \theta) \geq \widehat{C}_{\theta, B}\},
\]
where \(\widehat{C}_{\theta, B}\) is the cutoff determined by the partition, representing the adjusted \(\alpha\)-quantile of the distribution of \(\tau(\X_b, \theta_b)\) values within the partition containing \(\theta\). 
Then, if Assumption \ref{assumption:strong_consistency} holds, both \ourmethod{} and \ourmethodpp{} are asymptotically consistent, that is:
\[
\lim_{B \to \infty} \P\left(\theta \in \widehat{R}_B(\X) \mid \theta \right) = 1 - \alpha.
\]
\end{thm}

This shows that the partition constructed by \ourmethod{} and \ourmethodpp{} is designed so that local coverage closely approximates conditional coverage. Moreover, unlike conventional asymptotic approaches, which typically rely on the sample size of the observed dataset \( \mathbf{x} \), \( n \), our approach is independent of \( n \). Instead, it depends solely on the number of simulations \( B \), which can be scaled up given sufficient computational resources. This framework enables both \ourmethod{} and \ourmethodpp{} to achieve robust, distribution-free guarantees and asymptotically attain optimal conditional coverage as \( B \) increases.

In the appendix, we discuss the key results required for the theoretical framework, provide intuition for why our methods work, and present the proofs of the main results.

\section{Implementation Details and Tuning Parameters}
\label{sec:implementation}

To implement both \ourmethod{} and \ourmethodpp, we use the efficient decision tree and random forest implementations provided by \textit{scikit-learn}. Since our approaches use regression trees to partition $\Theta$, managing tree growth is crucial to prevent empty or redundant partition elements. This is achieved through both pre and post-pruning. The pre-pruning is executed in both methods by fixing the \texttt{min\_samples\_split} hyperparameter to a value, such as 100 or 300 samples, which enables us to obtain well-populated leaves resulting in more accurate estimations of $H(.|\theta)$ in our framework. 

In \ourmethod, we additionally apply post-pruning to remove extra leaves and nodes using cost-complexity pruning, balancing partition complexity with predictive performance by reducing the amount of less informative partition elements. Since post-pruning can reduce variability across regression trees and lessen the benefits of ensembling different partitions, it is not applied in \ourmethodpp. For all other hyperparameters of the decision tree and random forest algorithms, we use \textit{scikit-learn}'s default settings, except for the number of trees in the random forest, which we set to 200.


Our default setting for the tuning parameter $M$ is $M = K/2$, placing \ourmethodpp{} in a majority-vote regime. In this setup, $\theta'$ is included in $A(\theta)$ only if the majority of trees vote for $\theta'$ and $\theta$ to be in the same leaf. While this choice is intuitive and performs well in practice (see Section \ref{sec:lfi_problems}), it is not universally optimal, as some problems may require larger neighborhoods around $\theta$ due to the potential approximate non-invariance of statistics. 

To address this, we propose a straightforward grid-search algorithm for optimizing $M$ using a small, additional simulated validation grid. The main idea is to select $M$ from a fine grid between $0$ and $K$ that minimizes estimated deviation from conditional coverage. This is done by calculating coverage across the validation grid with a batch of statistics simulated at each fixed grid point and then computing the mean absolute error of the estimated coverage relative to the nominal confidence level $1 - \alpha$. Section \ref{sec:applications} provides further details on this process. Alternatively, coverage can be estimated using the LF2I diagnostic module \citep{dalmasso2021likelihood}, which also leverages an additional simulation set. Algorithm \ref{alg:tune_M} in Appendix \ref{appendix:alg_tune_M} outlines the tuning procedure for $M$.

\section{Applications}
\label{sec:applications}

In this section, we compare the coverage performance of \ourmethod \ and \ourmethodpp \ (both tuned and majority-votes versions) to the other state-of-the-art competing methods on tractable likelihood, likelihood-free, and nuisance parameter settings.

To assess performance, we use several statistic and likelihood/simulator combinations across multiple sample sizes $n$ and simulation budgets $B$ in both likelihood-based and likelihood-free scenarios. In the nuisance parameter setting, we benchmark our methods with two examples: one using likelihood-free inference and another based on likelihood. We set a common confidence level of $1 - \alpha = 0.95$ across all experiments. Comparisons are made against the following baseline methods:

\begin{itemize}
    \item \textbf{Gradient Boosting Quantile Regression (Boosting)}: implemented in the \textit{scikit-learn} library \citep{pedregosa2011scikit}, we use it to estimate $C_{\theta}$ through the $(1- \alpha)$ conditional quantile of the $\tau(\X, \theta)$ given $\theta$ \citep{dalmasso2020confidence}. We fix the maximum depth as 3, the tolerance for early stopping as 15 and the maximum iteration as 100. The remaining hyperparameters are fixed as \textit{scikit-learn}'s default. To control tree growth in each iteration, we limit the maximum depth to 3 (instead of the default, which is unlimited) and employ a weak learner ensemble approach, which is well-suited for boosting methods \citep{freund1997decision, hastie2009elements}.

    
    \item \textbf{Monte-Carlo (MC)}: implemented with an equally spaced grid over $\Theta$, we simulate  $n_{MC}$ statistics for each grid element and estimate $C_{\theta}$ using the $(1 - \alpha)$-quantile of these simulations at the closest grid point. For multi-dimensional $\Theta$, the grid is formed by combining equally spaced one-dimensional grids obtained along each coordinate of $\Theta$. To ensure comparability with other methods, the uni-dimensional grid size is set to $\left\lceil \frac{B}{n_{MC}}^{1/d} \right\rceil$, where $d = \dim \Theta$. This guarantees that the Monte-Carlo simulation budget in multi-dimensional problems will be close to the fixed budget $B$. We set in all cases $n_{MC} = 500$.
    \item \textbf{Asymptotic}: this approach relies on classic asymptotic results for certain test statistics. It estimates $C_{\theta}$ as the $(1-\alpha)$-quantile of the statistic's asymptotic invariant distribution. Since estimated statistics in the LFI setting lack invariant asymptotic approximations, this method is only applied for comparison in some scenarios with tractable likelihoods.
    \end{itemize}

To evaluate the conditional coverage performance of each approach in both likelihood-based and likelihood-free scenarios, we compute the Mean Absolute Error (MAE) of each method's conditional coverage concerning the confidence level $(1 - \alpha)$. We begin by estimating the conditional coverage of each method through the simulation of several statistics inside an evaluation grid. Let $\Theta_{\text{grid}}$ denote such a grid. We compute the coverage for each $\theta' \in \Theta_{\text{grid}}$ at level $\alpha$ for any cutoff estimation method $\widehat{C}$ as follows:
\begin{align}
\label{eq:cover_cutoff}
    \text{cover}_{\alpha}(\widehat{C}, \theta') := \frac{1}{n_{\text{sim}}} \sum_{i = 1}^{n_{\text{sim}}} \I \left( \tau(\theta', \X_{i}^{(\theta')}) \leq \widehat{C}_{\theta'} \right) \;,
\end{align}
where $n_{\text{sim}}$ represents the number of simulated statistics and $\X_i^{(\theta')}$ denotes the $i$-th observation simulated under $\theta'$. We then define the MAE of the method's conditional coverage relative to the nominal confidence level  $(1 - \alpha)$ as:
\begin{align}
\label{eq:MAE_cutoff}
    \text{MAE}(\widehat{C}, \alpha) := \frac{1}{|\Theta_{\text{grid}}|} \sum_{\theta' \in \Theta_{\text{grid}}} \left| \text{cover}_{\alpha}(\widehat{C}, \theta') - (1 - \alpha) \right| \;.
\end{align}
The MAE serves as an effective metric for quantifying the extent to which each method $\widehat{C}$ deviates from the nominal level of conditional coverage.

In the nuisance scenarios, we assess performance by measuring each method’s coverage deviation from the oracle coverage, defined as follows:
\begin{align}
    d_{\alpha}(\widehat{C}, C) := \frac{1}{|\Theta_{grid}|} \sum_{\theta' \in \Theta_{grid}} \left| \text{cover}_{\alpha}(\widehat{C}, \theta') - \text{cover}_{\alpha}(C, \theta') \right| \; ,
\end{align}
where $C$ is the oracle   cutoff specified in Eq. \ref{eq:oracleNuisance}. This oracle cutoff controls coverage across all parameters— main, and nuisance —by accounting for worst-case variations in the nuisance parameter. Consequently, we compare each method's coverage against the oracle rather than the nominal level, as the oracle typically over-covers for the parameters of interest to ensure reliable coverage across all nuisance parameters.

We replicate each experiment 50 times for likelihood-based, 30 times for likelihood-free, and 15 times for nuisance examples, and compute the average MAE or deviation from the oracle and its standard error. We identify the top methods in each case as those with the lowest average error, followed by any other methods with performance not significantly different from the best (determined by overlapping 95\% asymptotic confidence intervals for the MAE). All detailed MAE and standard errors visualizations are available in Appendix \ref{appendix:exp_results}.

\subsection{Examples with tractable  likelihood function}
\label{sec:tractable_likelihood}
In this section, we apply our methods to construct confidence sets on three standard statistical models. We always assume the data is $\X=(X_1,\ldots,X_n)$, where the $X_i$'s are i.i.d. To generate the data, we use the Normal model with fixed variance, the Gaussian mixture model (GMM), and the Lognormal model with both mean and scale as parameters. For the choice of $\tau$, we explore the following statistics: Likelihood ratio \citep{drton2009likelihood}; Kolmogorov-Smirnov \citep{marsaglia2003evaluating}; Frequentist Bayes Factor \citep{dalmasso2021likelihood} and E-value \citep{pereira1999evidence}. The detailed setup is described in Appendix \ref{appendix:tractable_ll}.

To compare all methods, we consider $n = 10, 20, 50, 100$ and $B= 1000, 5000, 10000, 15000$ for all combinations of statistic and likelihood. Figure \ref{fig:all_comparissons} provides a summary of each method's performance, while  Appendix \ref{appendix:trac_results} shows a detailed comparison of the models' performance for each of the statistics in each setting. 
Figure \ref{fig:all_comparissons}  highlights several key performance differences among the methods:




\begin{figure}[ht]
    \centering
    \includegraphics[scale = 0.165]{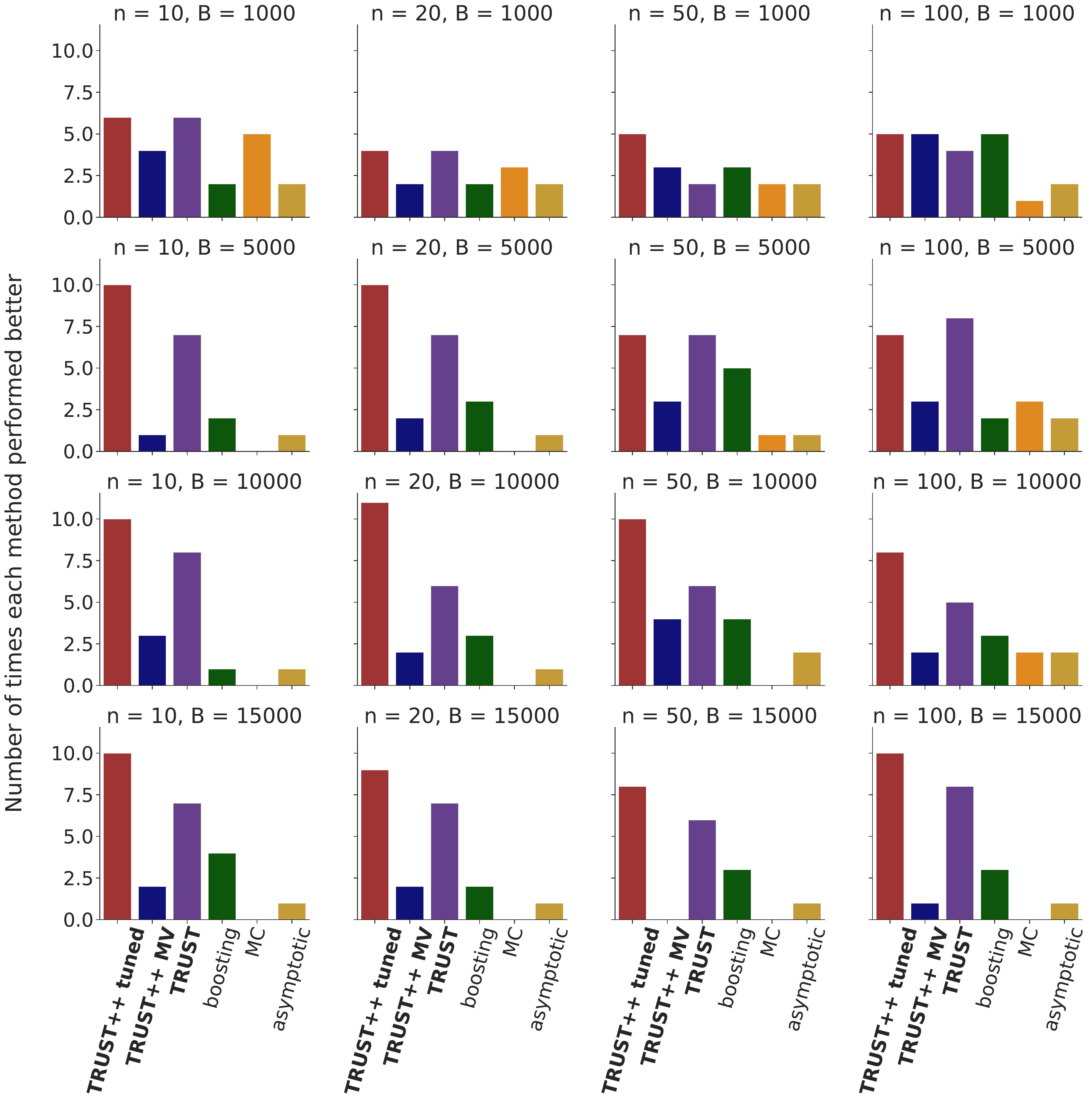}
    \caption{Tractable Problems: Frequency with which each method achieved the best performance (lowest significant MAE) across all statistic–likelihood combinations and varying $n$ and $B$. Our methods (in bold) consistently outperform competitors, with a clear advantage in smaller-sample settings ($n=10, 20$), where asymptotic approaches tend to underperform. For disaggregated results, see Appendix \ref{appendix:trac_results}.}
    \label{fig:all_comparissons}
\end{figure}

\begin{itemize}
    \item At $B = 1000$, there are no clear advantage among the methods for different sample sizes $n$, with \ourmethod{} and tuned \ourmethodpp{} showing competitive results in these scenarios. This balance may be attributed to the limited number of simulation samples, which reduces the ability to differentiate among methods regarding cutoff estimation and overall performance.
    \item Tuned \ourmethodpp{} consistently outperforms all competing methods across $B$ from $5000$ to $15000$ and for all $n$. 
    \item \ourmethod{} exhibits good performances for all $n$, showing a higher count of scenarios with superior outcomes across each combination of $n$ and $B$. 
    
    \item Both Monte-Carlo and Asymptotic methods present comparatively poor performance when the simulation budget exceeds $1000$.
    \item Although the boosting method demonstrates strong performance in certain scenarios and specific combinations of $n$ and $B$,  this method is generally outperformed by \ourmethod{} and tuned \ourmethodpp{} in the majority of cases. The detailed case-by-case analysis in Appendix \ref{appendix:trac_results} corroborates these findings.
    \item Examining the individual statistics (Appendix \ref{appendix:trac_results}), our methods show superior performance for the KS statistic, most notably \ourmethodpp{} (tuned) and \ourmethod{}. For the LR statistic, both the asymptotic approach and our methods achieve satisfactory results. Furthermore, our methods exhibit good performance for BFF and E-value, with boosting representing a competitive alternative in these settings.
    \end{itemize}

Based on all comparisons, we conclude that our framework consistently outperforms competing approaches in terms of conditional coverage for several settings.

\subsection{Likelihood-free Inference Problems}
\label{sec:lfi_problems}

We now apply our methods to construct confidence sets across five well-studied LFI benchmarks, where the likelihood is intractable or no analytical formula exists for certain statistics, requiring them to be estimated. In cases of intractable likelihood, we simulate the data $\X = (X_1, \dots, X_n)$ using a high-fidelity simulator $F_{\theta}$. Given a fixed $\theta$, we also assume that each $X_i$ is i.i.d. The simulators we use are the following: SLCP (Simple likelihood complex posterior), M/G/1, Weinberg \citep{hermans2021trust}; Two Moons \citep{greenberg2019automatic}; SIR \citep{lueckmann2021benchmarking}. We consider the folllowing choices of $\tau$: the E-Value, BFF, and Waldo \citep{masserano2023simulator}. A detailed description of them can be found in the Appendix \ref{appendix:intractable_ll}.

Since an exact posterior $f(\theta|\x)$ is unavailable for all statistics in the likelihood-free inference  context, we estimate posterior quantities using a Neural Posterior Estimator based on normalizing flows \citep{rezende2015variational,hermans2021trust}. This method approximates the posterior $f(\theta|\x)$ through an amortized estimator $\hat{f}_{\psi}(\theta|\x)$, constructed with neural network-based bijective transformations parameterized by $\psi$ \citep{rezende2015variational}. The estimator is trained on a simulated sample \(\{(\theta_1, \X_1), \dots, (\theta_B, \X_B)\}\) generated from each model or simulator. Details of normalizing flow architecture, implementation, and simulation budgets are given in Appendix \ref{appendix:experiment_details}.

For each statistic and simulator, we evaluate all methods across combinations of $n \in \{ 1, 5, 10, 20 \}$ and $B\in \{ 10, 15, 20, 30\} \times 10^3$, increasing the simulation budgets relative to Section \ref{sec:tractable_likelihood} to account for the greater complexity of these problems. Figure \ref{fig:all_real_comparisons} presents a summary of the overall performance of each method, while Appendix \ref{appendix:intrac_results} provides a detailed comparison of coverage results, disaggregated by statistic and data-generating process.
The results again reveal several performance distinctions among the methods:
\begin{figure}[ht]
    \centering
    \includegraphics[scale = 0.165]{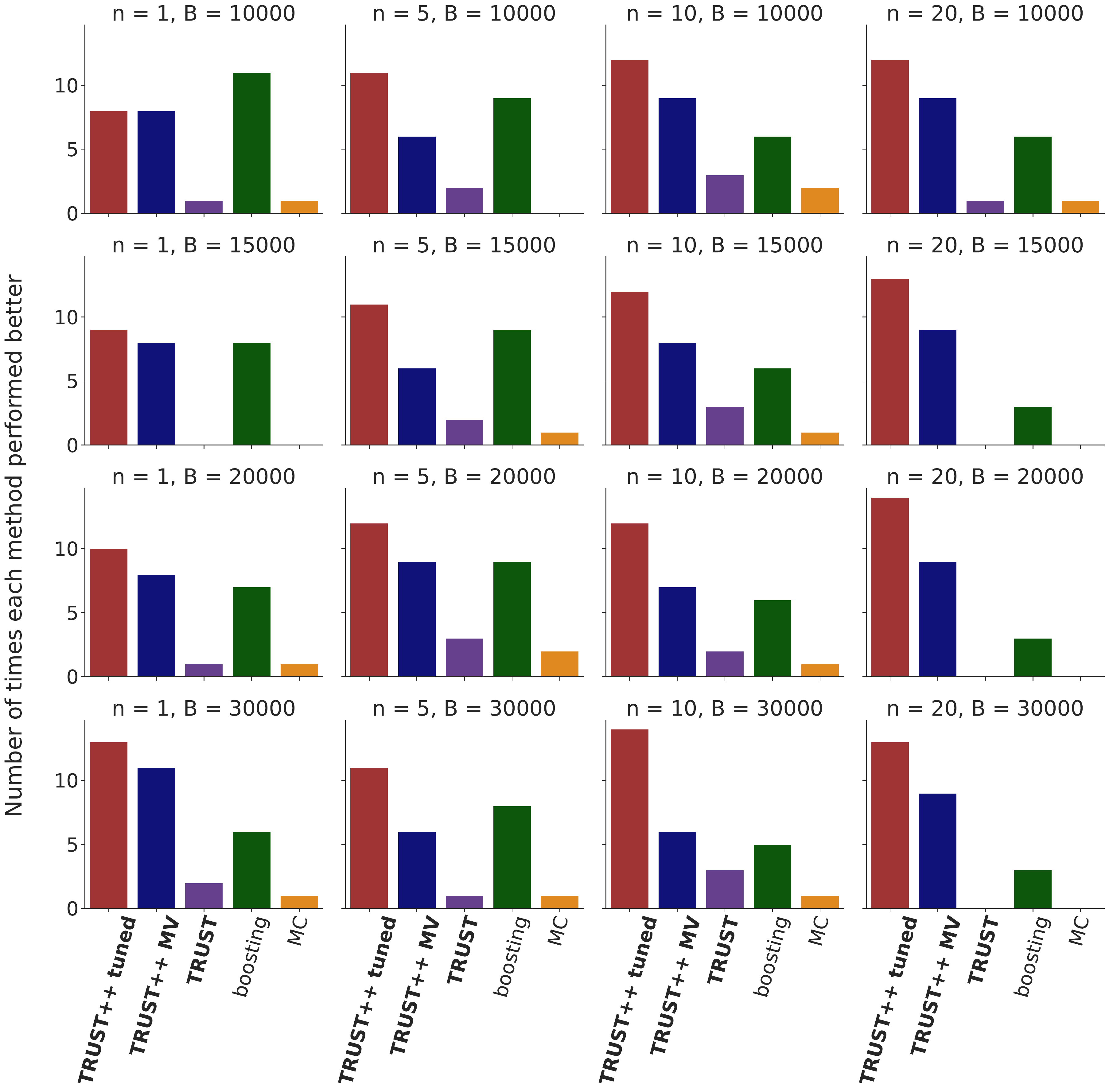}
    \caption{Likelihood-free Inference Problems: Frequency with which each method achieved the best result for each statistic–simulator combination across varying sample sizes $n$ and simulation budgets $B$. Our methods (in bold) show consistent superiority when $B > 10{,}000$ and remain competitive even at $B = 10{,}000$. For disaggregated results, see Appendix \ref{appendix:intrac_results}.}
    \label{fig:all_real_comparisons}
\end{figure}

\begin{itemize}
\item For $n = 1$ and $5$ with simulation budgets below $30000$, Figure \ref{fig:all_real_comparisons} shows a relatively balanced performance between tuned \ourmethodpp{}, \ourmethodpp{} MV, and boosting. However, in all scenarios with $n = 10$ or $20$, tuned \ourmethodpp{} consistently achieves superior results, with \ourmethodpp{} MV performing strongly and both outperforming boosting. This suggests that our Breiman distance approach for computing local cutoffs adapts effectively to likelihood-free settings as sample sizes and simulation budgets increase. 
\item Figure \ref{fig:all_real_comparisons} indicates that tuned \ourmethodpp{} outperforms competing methods in nearly all combinations (15 out of 16). 
Again, tuned \ourmethodpp{}  exhibits a high count of combinations compared to all other approaches, mirroring its  performance observed in tractable experiments.
\item The Monte Carlo method performs poorly across nearly all scenarios, showing weaker coverage control than other methods. 
\item The boosting method demonstrates competitive coverage control and consistent performance in several scenarios, notably outperforming in the specific case of $n = 1$ and $B = 10000$. However, in most cases, it is outperformed by both tuned \ourmethodpp{} and \ourmethodpp{} MV. 


\item Despite exhibiting poorer coverage control than boosting and tuned \ourmethodpp{}, \ourmethodpp{} MV still shows good performance across all scenarios. 
\item In contrast to the results seen in tractable experiments, \ourmethod{} demonstrates subpar performance and coverage control compared to both boosting and \ourmethodpp{} in likelihood-free scenarios. This indicates that the \ourmethod{} partition requires enhancements for likelihood-free contexts, a need effectively addressed through the ensembling approach used in \ourmethodpp{}.

\item Appendix \ref{appendix:intrac_results} further corroborates the effectiveness of our methods, showing that both \ourmethodpp{} tuned and \ourmethodpp{} MV perform consistently well across all statistics. Boosting also remains a good competitor, particularly for Waldo and E-value. Notably, in most scenarios where either \ourmethodpp{} tuned or \ourmethodpp{} MV attains the best performance, the other method performs within a narrow margin, effectively placing both among the top-performing approaches.
\end{itemize}

Figure \ref{fig:CI_e_value_two_moons} further illustrates the application of each method, showcasing the confidence regions alongside \ourmethodpp{}'s uncertainty quantification for a specific realization of the two moons example. Both \ourmethodpp{} and boosting closely approximate the oracle region, whereas the Monte Carlo method significantly underestimates the region's size. While \ourmethodpp{} and boosting produce similar confidence regions, \ourmethodpp{} distinguishes itself by offering an additional layer of uncertainty quantification. This uncertainty layer reveals insights into oracle information that may not be fully captured by the region estimators, likely due to constraints in the simulation budget. 

\begin{figure}[h]
    \centering
    \includegraphics[width=0.65\linewidth]{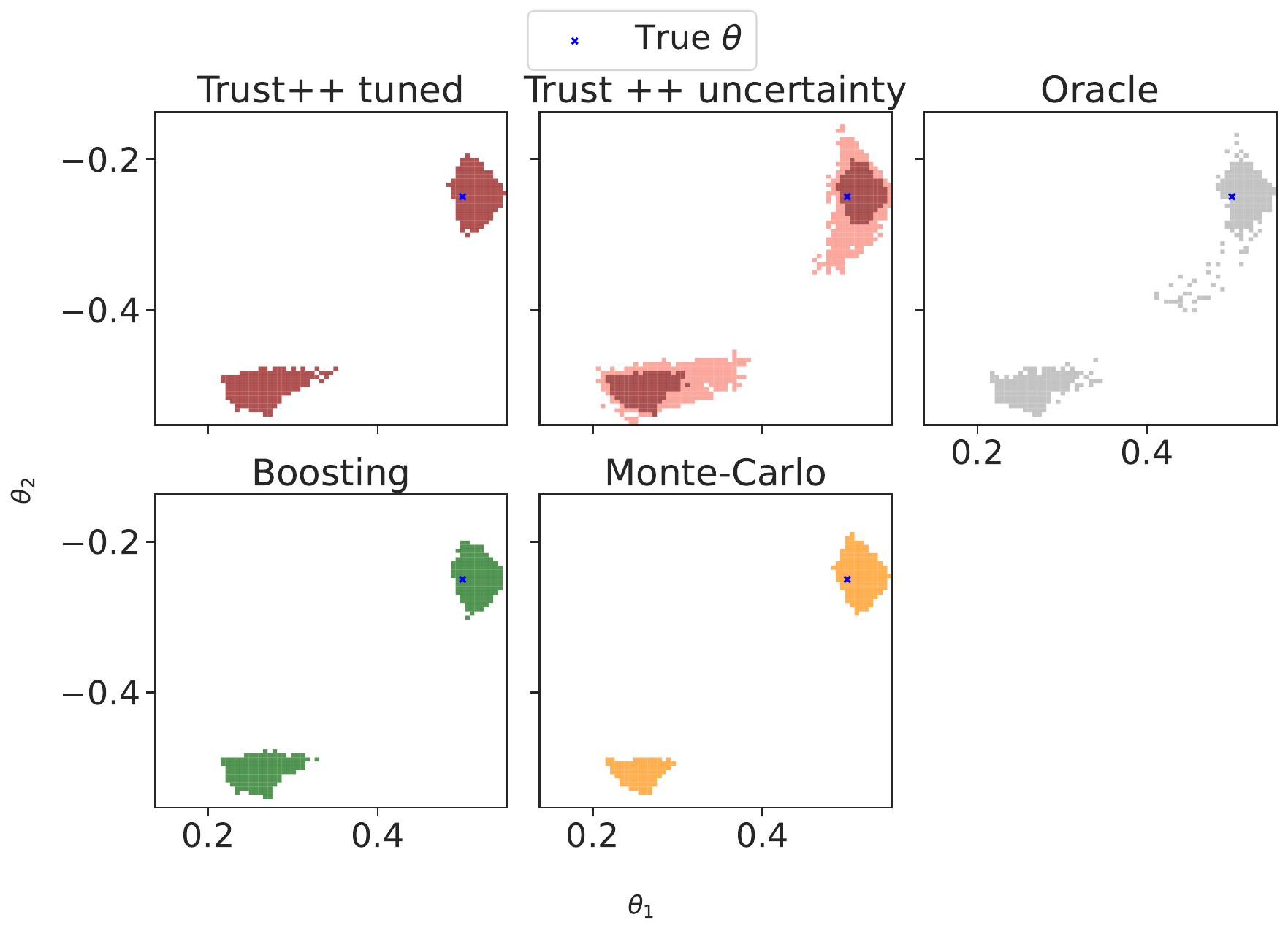}
    \caption{95\% confidence regions and \ourmethodpp{} uncertainty for the two-moons example using the e-value statistic ($n=20$, $B=15{,}000$). Maroon areas denote high confidence, while salmon areas indicate uncertainty. \ourmethodpp{} closely matches the oracle set and highlights where additional simulations could reduce uncertainty.}
    \label{fig:CI_e_value_two_moons}
\end{figure}

\subsection{Examples of nuisance-parameter problems}
\label{sec:nuisanceExamples}

In this section, we demonstrate the application of both of our methods in scenarios involving nuisance parameters. We begin with a Poisson counting experiment, originally introduced by \cite{dalmasso2021likelihood} as an example from high-energy physics, where the BFF statistic is non-invariant, requiring both estimation and marginalization. Next, we tackle a classic inference problem: estimating parameter intervals within a gamma generalized linear model (GLM) under the presence of nuisance parameters and limited sample sizes. We compare our methods to the commonly used asymptotic approach, employing a likelihood ratio statistic that is marginalized over the nuisance parameter \citep{mccullagh2019generalized}.

\subsubsection{Poisson Counting Experiment}
We consider the observed data $\X = (N_b, N_s)$  which consists of two counts where $N_{b} \sim \text{Pois}(\nu \tau b)$ and $N_{s} \sim \text{Pois}(\nu b + \mu s)$. In this example, $b$, $s$, and $\tau$ are fixed hyperparameters, while $\mu$ is the parameter of interest and $\nu$ is the nuisance parameter. We define $\boldsymbol{\theta} = (\mu, \nu) \in \Theta = (0, 5) \times (0,1.5)$ and fix the hyper-parameters as $s = 15$, $b = 70$ and $\tau = 1$  to avoid the Gaussian limiting regime for Poisson distributions. We also take the uniform distribution over the parameter space as prior. To derive a statistic $\tau(\X, \mu)$ in this case, in this context, we use a marginalized BFF statistic based on the posterior $f(\mu|\x)$, estimated through a normalizing flows approach (details in Appendix \ref{appendix:experiment_details}). We set the simulation budget to $B = 10 \times 10^3$ and $n = 1$ to fit all methods.

Although the BFF statistic is marginalized over nuisance parameters, it may still depend on their values. Therefore, we fit all methods considering the full parameter $(\mu, \nu)$ and compute the cutoff for $\mu$ through Eq. \ref{eq:cNuisance}. Our approaches facilitate this computation by using the tree structure and achieve better results, as shown in Table \ref{tab:comparisson_poisson_nuisance}.
\begin{table}[!ht]
\centering
\caption{Mean absolute coverage deviation from the oracle for each method in the Poisson counting example. The average across 15 runs is reported along with twice its standard error. \ourmethodpp{} demonstrates exceptional performance.}
\label{tab:comparisson_poisson_nuisance}
\begin{tabular}{ccc}
\hline
Methods   & $d_{\alpha}$    & $SE \cdot 2$      \\
\hline
\ourmethod{}     & 0.0369 & $0.14 \cdot 10^{-4}$  \\
\ourmethodpp{}   & 0.0041 & $0.15 \cdot 10^{-4}$  \\
Boosting  & 0.0084 & $0.12 \cdot 10^{-4}$ \\
MC        & 0.1460 & $0.11 \cdot 10^{-4}$ \\
\hline
\end{tabular}
\end{table}

Table \ref{tab:comparisson_poisson_nuisance} highlights the strong performance of \ourmethodpp{}, with coverage closely aligning with the oracle, emphasizing its adaptability to challenges involving nuisance parameters. Also, we observe that the Monte Carlo method shows a substantial deviation from the oracle, further illustrating its limitations in adapting effectively to different inference scenarios. 

While \ourmethodpp{} achieves the best results, boosting also performs well, providing coverage that is similarly close to the oracle's. Figure \ref{fig:oracle_diff_poisson} details the differences between the methods. We observe that \ourmethodpp{} effectively controls coverage deviation across all combinations of $(\mu, \nu)$. In contrast, boosting exhibits areas of higher deviation, particularly around $\mu = 2$ for all values of $\nu$ and around $\mu = 3$ when $\nu$ is low. This demonstrates \ourmethodpp{}'s accuracy in approximating the oracle region throughout the parameter space, indicating its effectiveness in estimating the cutoff $C_\mu$ defined in Eq. \ref{eq:cNuisance}.

\begin{figure}[!ht]
    \centering
    \includegraphics[width=0.8\linewidth]{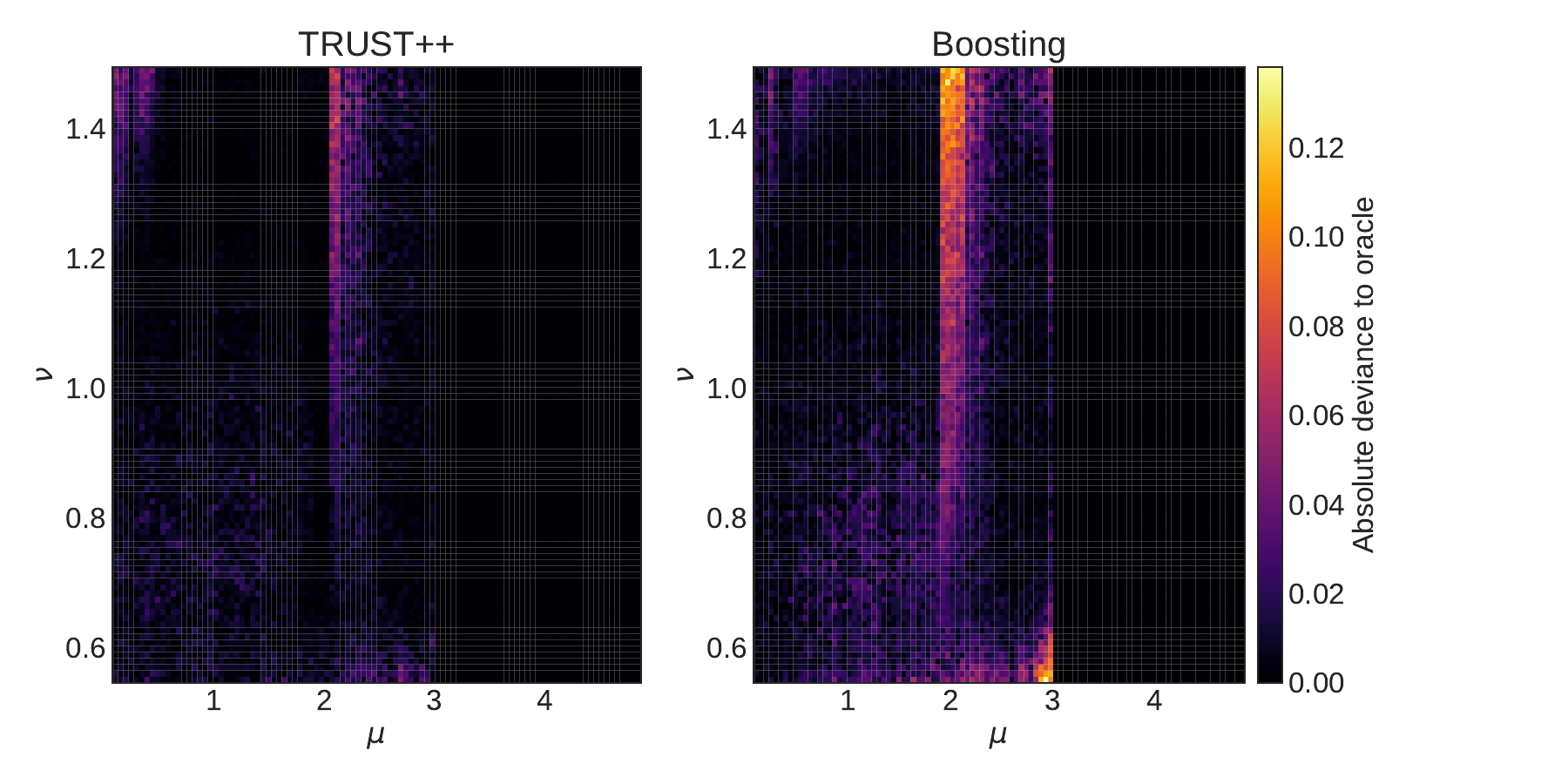}
    \caption{Comparison of absolute deviations from the oracle for each pair of parameters between \ourmethodpp{} and boosting. Notably, around $\mu = 2$, boosting exhibits significant deviations from the oracle, whereas \ourmethodpp{} effectively controls these deviations, keeping them below $0.1$.}
    \label{fig:oracle_diff_poisson}
\end{figure}

\subsubsection{Gamma GLM experiment} Consider i.i.d data $Y_i \in \mathbb{R^{+}}$ generated from a gamma generalized linear model (GLM) with fixed covariates $\X_i = (X_{i,1}, X_{i,2}) \in (-1, 1)^2$:
\begin{align}
\label{eq:likelihood_GLM}
    Y_i \sim \text{Gamma}\left(1/\phi, \phi \cdot \exp{\{ \beta_0 + \beta_1 X_{i, 1} + \beta_2 X_{i, 2} \}}  \right) \; ,
\end{align}
 where $\phi$ is the dispersion parameter and $\beta_0, \beta_1$ and $\beta_2$ are coefficients for the linear predictor, with $\phi \in (0,1.75)$ and $\boldsymbol{\beta} = (\beta_0, \beta_1, \beta_2) \in \R^{3}$. The expected value of each $Y_i$ is given by:
\begin{align*}
\E[Y_i] = \exp{\{ \beta_0 + \beta_1 X_1 + \beta_2 X_2\}} \;.
\end{align*}
A standard problem in this setting is to derive a valid confidence interval for a single parameter $\beta_j$. A widely used approach is to construct an asymptotic confidence interval based on the marginalized likelihood ratio statistic:
$
    LR(\y, \beta_j) = \log  \mathcal{L}(\y; \boldsymbol{\widehat{\beta}}_{-j}, \beta_j) - \log \mathcal{L}(\y;\boldsymbol{\widehat{\beta}})$,
where $\boldsymbol{\widehat{\beta}}$ is the MLE, $\boldsymbol{\widehat{\beta}}_{-j}$ the MLE restricted under a fixed $\beta_j$ and $\mathcal{L}(.;\boldsymbol{\beta})$ denotes the GLM likelihood as defined in Eq. \ref{eq:likelihood_GLM}. Using this statistic, an asymptotic confidence interval can be easily derived by referencing the $\chi^2_1$ distribution to compute $C_{\beta_j}$ in a invariant way. 

The limitation of this approach is that it relies on large sample sizes to be effective; for small samples, it may fail to produce a valid confidence interval, as illustrated in Figure \ref{fig:example_CI_and_prob_coverage_GLM}(b). This is because the marginalized likelihood ratio statistic still depends on the remaining nuisance parameters,  $\boldsymbol{\beta}_{-j}$ and $\phi$ in low-sample regimes. This limitation can be addressed with our approach by computing $C_{\beta_j}$ using Eq. \ref{eq:cNuisance} in a scalable manner. To apply all simulation-based approaches, we assume the priors $\boldsymbol{\beta} \sim N(0, 4) \cdot N(0,1)^2$ and $\phi \sim \text{Truncated Exponential}(1, 1.75)$. For method comparison, we set the parameter of interest to $j = 1$, the sample size to $n = 50$ and the simulation budget to $B = 10 \cdot 10^3$, with the covariate matrix $\X$  fixed according to values generated independently from $U(-1,1)^2$. Table \ref{tab:comparisson_glm_nuisance} shows the mean deviation from the oracle coverage for each method, while Figure \ref{fig:example_CI_and_prob_coverage_GLM}  illustrates the probability of coverage and confidence intervals for specific values of $\beta_j$ and $\Y$ samples.
\begin{table}[h]
\centering
\caption{Mean absolute coverage deviation from the oracle for each method in the GLM example. The average across 15 runs is reported along with twice its standard error. \ourmethodpp{} demonstrates superior performance compared to all competing methods.}
\label{tab:comparisson_glm_nuisance}
\begin{tabular}{ccc}
\hline
Methods      & $d_{\alpha}$      & $SE \cdot 2$      \\ \hline
\ourmethod       & $0.0312$  & $0.176 \cdot 10^{-3}$  \\ 
\ourmethodpp{}      & $0.0157$  & $0.163 \cdot 10^{-3}$  \\
Boosting     & $0.0234$  & $0.171 \cdot 10^{-3}$ \\ 
MC           & $0.0207$  & $0.157 \cdot 10^{-3}$ \\
Asymptotic   & $0.0327$  & $0.171 \cdot 10^{-3}$  \\ \hline
\end{tabular}
\end{table}

By Table \ref{tab:comparisson_glm_nuisance} we notice that \ourmethodpp{} achieves the best performance, with coverage closely matching the oracle and outperforming Monte Carlo and boosting by a large margin. Figure \ref{fig:example_CI_and_prob_coverage_GLM} illustrate this proximity for specific values of $\beta_j$ and observed sample values $\Y$. Although \ourmethodpp{} exhibits overcoverage relative to the nominal level $1-\alpha$ and produces larger confidence intervals compared to other methods, it closely emulates the oracle’s behavior, making it more valid than the competing approaches. Additionally, \ourmethodpp{}'s uncertainty bar indicates that for nearly all parameters within the confidence interval, we can be confident they truly lie within the interval, reflecting a well-estimated range.  Furthermore, both Table \ref{tab:comparisson_glm_nuisance} and Figure \ref{fig:example_CI_and_prob_coverage_GLM} reveal that the asymptotic approach not only deviates more from the oracle but also undercovers and underestimates the confidence interval length. This underscores the limitations of asymptotic methods for small sample sizes and nuisance parameter, even in problems with tractable likelihoods.

\section{Final Remarks}
\label{sec:final_remarks}

This paper introduces novel methods, \ourmethod \ and \ourmethodpp, for the distribution-free calibration of statistical confidence sets, ensuring both finite-sample local coverage and asymptotic conditional coverage. By leveraging tools from conformal inference, we create robust confidence sets that extend the applicability of traditional conformal techniques, typically used for generating prediction intervals, to the domain of statistical inference. This adaptation allows us to maintain the desired coverage properties even in challenging settings. In practice, this translates to superior performances across a variety of scenarios, particularly in cases involving small sample sizes $n$ and larger simulation budgets $B$. 

Furthermore, our methods effectively address inference settings with nuisance parameters, a notable challenge for other approaches, especially in likelihood-free inference. Unlike existing techniques, \ourmethod \ and \ourmethodpp \ also enable robust uncertainty quantification about the oracle confidence sets, providing valuable insights into whether additional simulated data is needed to achieve more reliable estimates. Additionally, we can leverage the third branch of LF2I \citep{dalmasso2020confidence,dalmasso2021likelihood} to test for the exact conditional coverage of our confidence sets.

Theorem \ref{thm:asympt_conditional} demonstrates that both \ourmethod\ and \ourmethodpp\ are consistent estimators for quantile regression. While we applied these methods specifically to quantile-regress the conformal score on \(\theta\), their utility extends far beyond this particular setting. They can be applied in a variety of contexts, such as constructing predictive intervals in a prediction setting, including for multivariate \(Y\). This capability allows for the generation of label-conditional predictive sets that aim to control \( \P(Y \in \R(\X) \mid Y = y) \), ensuring coverage conditional on the label.

To our knowledge, no other conformal methods aim to control label-conditional coverage for continuous outcomes \(Y\), making this a novel contribution to the field. This work paves the way for future research to explore label-conditional inference in complex, continuous and multivariate-label settings, potentially improving the  reliability of prediction intervals in domains where conditional control is critical.


\subsubsection*{Acknowledgments}
L.M.C.C is grateful for the fellowship provided by São Paulo Research Foundation (FAPESP), grant 2022/08579-7. R. I. is grateful for the financial support of FAPESP (grants 2019/11321-9 and 2023/07068-1) and 
CNPq (grants 422705/2021-7 and 305065/2023-8). R. B. S. produced this work as part of the activities of Fundação de Amparo
à Pesquisa do Estado de São Paulo Research, Innovation and Dissemination Center for Neuromathematics (grant 2013/07699-0). The authors are also grateful to Rodrigo F. L. Lassance and Ann B. Lee for their suggestions and insightful discussions.

\bibliography{main}

\begin{thebibliography}{63}
\providecommand{\natexlab}[1]{#1}
\providecommand{\url}[1]{\texttt{#1}}
\expandafter\ifx\csname urlstyle\endcsname\relax
  \providecommand{\doi}[1]{doi: #1}\else
  \providecommand{\doi}{doi: \begingroup \urlstyle{rm}\Url}\fi

\bibitem[Algeri et~al.(2019)Algeri, Aalbers, Mor{\aa}, and Conrad]{algeri2019searching}
Sara Algeri, Jelle Aalbers, Knut~Dundas Mor{\aa}, and Jan Conrad.
\newblock Searching for new physics with profile likelihoods: Wilks and beyond.
\newblock \emph{arXiv preprint arXiv:1911.10237}, 2019.

\bibitem[Angelopoulos et~al.(2023)Angelopoulos, Bates, et~al.]{angelopoulos2023conformal}
Anastasios~N Angelopoulos, Stephen Bates, et~al.
\newblock Conformal prediction: A gentle introduction.
\newblock \emph{Foundations and Trends{\textregistered} in Machine Learning}, 16\penalty0 (4):\penalty0 494--591, 2023.

\bibitem[Baragatti et~al.(2024)Baragatti, Cloez, M{\'e}tivier, and Sanchez]{baragatti2024approximate}
Meili Baragatti, Bertrand Cloez, David M{\'e}tivier, and Isabelle Sanchez.
\newblock Approximate bayesian computation with deep learning and conformal prediction.
\newblock \emph{arXiv preprint arXiv:2406.04874}, 2024.

\bibitem[Berg(2004)]{berg2004no}
Nathan Berg.
\newblock No-decision classification: an alternative to testing for statistical significance.
\newblock \emph{the Journal of socio-Economics}, 33\penalty0 (5):\penalty0 631--650, 2004.

\bibitem[Biau et~al.(2008)Biau, Devroye, and Lugosi]{consistency_gabor}
G\'{e}rard Biau, Luc Devroye, and G\'{a}bor Lugosi.
\newblock Consistency of random forests and other averaging classifiers.
\newblock \emph{J. Mach. Learn. Res.}, 9:\penalty0 2015–2033, June 2008.
\newblock ISSN 1532-4435.

\bibitem[Blum and Fran{\c{c}}ois(2010)]{blum2010non}
Michael~GB Blum and Olivier Fran{\c{c}}ois.
\newblock Non-linear regression models for approximate bayesian computation.
\newblock \emph{Statistics and computing}, 20:\penalty0 63--73, 2010.

\bibitem[Bostr{\"o}m and Johansson(2020)]{bostrom2020mondrian}
Henrik Bostr{\"o}m and Ulf Johansson.
\newblock Mondrian conformal regressors.
\newblock In \emph{Conformal and Probabilistic Prediction and Applications}, pages 114--133. PMLR, 2020.

\bibitem[Bostr{\"{o}}m et~al.(2021)Bostr{\"{o}}m, Johansson, and L{\"{o}}fstr{\"{o}}m]{Bostroem2021}
Henrik Bostr{\"{o}}m, Ulf Johansson, and Tuwe L{\"{o}}fstr{\"{o}}m.
\newblock Mondrian conformal predictive distributions.
\newblock In Lars Carlsson, Zhiyuan Luo, Giovanni Cherubin, and Khuong~An Nguyen, editors, \emph{Conformal and Probabilistic Prediction and Applications, 8-10 September 2021, Virtual Event}, volume 152 of \emph{Proceedings of Machine Learning Research}, pages 24--38. {PMLR}, 2021.
\newblock URL \url{https://proceedings.mlr.press/v152/bostrom21a.html}.

\bibitem[Brehmer et~al.(2020)Brehmer, Louppe, Pavez, and Cranmer]{brehmer2020mining}
Johann Brehmer, Gilles Louppe, Juan Pavez, and Kyle Cranmer.
\newblock Mining gold from implicit models to improve likelihood-free inference.
\newblock \emph{Proceedings of the National Academy of Sciences}, 117\penalty0 (10):\penalty0 5242--5249, 2020.

\bibitem[Breiman(2001)]{breiman2001random}
Leo Breiman.
\newblock Random forests.
\newblock \emph{Machine learning}, 45:\penalty0 5--32, 2001.

\bibitem[Cabezas et~al.(2025)Cabezas, Otto, Izbicki, and Stern]{cabezas2025regression}
Luben~MC Cabezas, Mateus~P Otto, Rafael Izbicki, and Rafael~B Stern.
\newblock Regression trees for fast and adaptive prediction intervals.
\newblock \emph{Information Sciences}, 686:\penalty0 121369, 2025.

\bibitem[Casella and Berger(2024)]{casella2024statistical}
George Casella and Roger Berger.
\newblock \emph{Statistical inference}.
\newblock CRC Press, 2024.

\bibitem[Chen and Li(2009)]{chen2009hypothesis}
Jiahua Chen and Pengfei Li.
\newblock Hypothesis test for normal mixture models: The em approach.
\newblock 2009.

\bibitem[Cranmer et~al.(2020)Cranmer, Brehmer, and Louppe]{cranmer2020frontier}
Kyle Cranmer, Johann Brehmer, and Gilles Louppe.
\newblock The frontier of simulation-based inference.
\newblock \emph{Proceedings of the National Academy of Sciences}, 117\penalty0 (48):\penalty0 30055--30062, 2020.

\bibitem[Dalmasso et~al.(2020)Dalmasso, Izbicki, and Lee]{dalmasso2020confidence}
Niccolo Dalmasso, Rafael Izbicki, and Ann Lee.
\newblock Confidence sets and hypothesis testing in a likelihood-free inference setting.
\newblock In \emph{International Conference on Machine Learning}, pages 2323--2334. PMLR, 2020.

\bibitem[Dalmasso et~al.(2021)Dalmasso, Masserano, Zhao, Izbicki, and Lee]{dalmasso2021likelihood}
Niccol{\`o} Dalmasso, Luca Masserano, David Zhao, Rafael Izbicki, and Ann~B Lee.
\newblock Likelihood-free frequentist inference: Bridging classical statistics and machine learning in simulator-based inference.
\newblock \emph{arXiv preprint arXiv:2107.03920}, 2021.

\bibitem[de~B~Pereira et~al.(2008)de~B~Pereira, Stern, and Wechsler]{pereira2008can}
Carlos~A de~B~Pereira, Julio~Michael Stern, and Sergio Wechsler.
\newblock Can a significance test be genuinely bayesian?
\newblock \emph{Bayesian Analysis}, 3\penalty0 (1):\penalty0 79--100, 2008.

\bibitem[DeGroot and Schervish(2012)]{degroot2012probability}
Morris~H DeGroot and Mark~J Schervish.
\newblock Probability and statistics.[sl].
\newblock \emph{Pearson Education}, 19:\penalty0 33, 2012.

\bibitem[Dheur et~al.(2024)Dheur, Bosser, Izbicki, and Taieb]{dheur2024distribution}
Victor Dheur, Tanguy Bosser, Rafael Izbicki, and Souhaib~Ben Taieb.
\newblock Distribution-free conformal joint prediction regions for neural marked temporal point processes.
\newblock \emph{arXiv preprint arXiv:2401.04612}, 2024.

\bibitem[Ding et~al.(2023)Ding, Angelopoulos, Bates, Jordan, and Tibshirani]{ding2023class}
Tiffany Ding, Anastasios Angelopoulos, Stephen Bates, Michael Jordan, and Ryan~J Tibshirani.
\newblock Class-conditional conformal prediction with many classes.
\newblock \emph{Advances in neural information processing systems}, 36:\penalty0 64555--64576, 2023.

\bibitem[Diniz et~al.(2012)Diniz, Pereira, Polpo, Stern, and Wechsler]{diniz2012relationship}
Marcio Diniz, Carlos~AB Pereira, Adriano Polpo, Julio~M Stern, and Sergio Wechsler.
\newblock Relationship between bayesian and frequentist significance indices.
\newblock \emph{International Journal for Uncertainty Quantification}, 2\penalty0 (2), 2012.

\bibitem[Drton(2009)]{drton2009likelihood}
Mathias Drton.
\newblock Likelihood ratio tests and singularities.
\newblock 2009.

\bibitem[Durkan et~al.(2019)Durkan, Bekasov, Murray, and Papamakarios]{durkan2019neural}
Conor Durkan, Artur Bekasov, Iain Murray, and George Papamakarios.
\newblock Neural spline flows.
\newblock \emph{Advances in neural information processing systems}, 32, 2019.

\bibitem[Esteves et~al.(2016)Esteves, Izbicki, Stern, and Stern]{esteves2016logical}
Lu{\'\i}s~G Esteves, Rafael Izbicki, Julio~M Stern, and Rafael~B Stern.
\newblock The logical consistency of simultaneous agnostic hypothesis tests.
\newblock \emph{Entropy}, 18\penalty0 (7):\penalty0 256, 2016.

\bibitem[Freund and Schapire(1997)]{freund1997decision}
Yoav Freund and Robert~E Schapire.
\newblock A decision-theoretic generalization of on-line learning and an application to boosting.
\newblock \emph{Journal of computer and system sciences}, 55\penalty0 (1):\penalty0 119--139, 1997.

\bibitem[Greenberg et~al.(2019)Greenberg, Nonnenmacher, and Macke]{greenberg2019automatic}
David Greenberg, Marcel Nonnenmacher, and Jakob Macke.
\newblock Automatic posterior transformation for likelihood-free inference.
\newblock In \emph{International Conference on Machine Learning}, pages 2404--2414. PMLR, 2019.

\bibitem[Guan(2023)]{guan2023localized}
Leying Guan.
\newblock Localized conformal prediction: A generalized inference framework for conformal prediction.
\newblock \emph{Biometrika}, 110\penalty0 (1):\penalty0 33--50, 2023.

\bibitem[Hahn and Meeker(2011)]{hahn2011statistical}
Gerald~J Hahn and William~Q Meeker.
\newblock \emph{Statistical intervals: a guide for practitioners}, volume~92.
\newblock John Wiley \& Sons, 2011.

\bibitem[Hastie et~al.(2009)Hastie, Tibshirani, Friedman, and Friedman]{hastie2009elements}
Trevor Hastie, Robert Tibshirani, Jerome~H Friedman, and Jerome~H Friedman.
\newblock \emph{The elements of statistical learning: data mining, inference, and prediction}, volume~2.
\newblock Springer, 2009.

\bibitem[Hermans et~al.(2021)Hermans, Delaunoy, Rozet, Wehenkel, Begy, and Louppe]{hermans2021trust}
Joeri Hermans, Arnaud Delaunoy, Fran{\c{c}}ois Rozet, Antoine Wehenkel, Volodimir Begy, and Gilles Louppe.
\newblock A trust crisis in simulation-based inference? your posterior approximations can be unfaithful.
\newblock \emph{arXiv preprint arXiv:2110.06581}, 2021.

\bibitem[Hore and Barber(2023)]{hore2023conformal}
Rohan Hore and Rina~Foygel Barber.
\newblock Conformal prediction with local weights: randomization enables local guarantees.
\newblock \emph{arXiv preprint arXiv:2310.07850}, 2023.

\bibitem[Izbicki et~al.(2014)Izbicki, Lee, and Schafer]{izbickiLeeSchafer}
R.~Izbicki, A.B. Lee, and C.M. Schafer.
\newblock High-dimensional density ratio estimation with extensions to approximate likelihood computation.
\newblock \emph{Journal of Machine Learning Research (AISTATS Track)}, pages 420--429, 2014.

\bibitem[Izbicki et~al.(2019)Izbicki, Lee, and Pospisil]{izbicki2019abc}
Rafael Izbicki, Ann~B Lee, and Taylor Pospisil.
\newblock {ABC--CDE}: Toward approximate bayesian computation with complex high-dimensional data and limited simulations.
\newblock \emph{Journal of Computational and Graphical Statistics}, 28\penalty0 (3):\penalty0 481--492, 2019.

\bibitem[Izbicki et~al.(2020)Izbicki, Shimizu, and Stern]{izbicki2020flexible}
Rafael Izbicki, Gilson Shimizu, and Rafael Stern.
\newblock Flexible distribution-free conditional predictive bands using density estimators.
\newblock In \emph{International Conference on Artificial Intelligence and Statistics}, pages 3068--3077. PMLR, 2020.

\bibitem[Izbicki et~al.(2022)Izbicki, Shimizu, and Stern]{izbicki2022cd}
Rafael Izbicki, Gilson Shimizu, and Rafael~B Stern.
\newblock Cd-split and hpd-split: Efficient conformal regions in high dimensions.
\newblock \emph{Journal of Machine Learning Research}, 23\penalty0 (87):\penalty0 1--32, 2022.

\bibitem[Izbicki et~al.(2025)Izbicki, Cabezas, Colugnatti, Lassance, de~Souza, and Stern]{Izbicki2025REACT}
Rafael Izbicki, Luben M.~C. Cabezas, Fernando A.~B. Colugnatti, Rodrigo F.~L. Lassance, Altay A.~L. de~Souza, and Rafael~B. Stern.
\newblock React to nhst: Sensible conclusions for meaningful hypotheses.
\newblock \emph{The Quantitative Methods for Psychology}, 21\penalty0 (2):\penalty0 43--66, 2025.
\newblock \doi{10.20982/tqmp.21.2.p043}.
\newblock URL \url{http://www.tqmp.org/RegularArticles/vol21-2/p043/p043.pdf}.

\bibitem[Jonkers et~al.(2024)Jonkers, Van~Wallendael, Duchateau, and Van~Hoecke]{jonkers2024cpd}
Jef Jonkers, Glenn Van~Wallendael, Luc Duchateau, and Sofie Van~Hoecke.
\newblock Conformal predictive systems under covariate shift.
\newblock \emph{arXiv preprint arXiv:2404.15018}, 2024.

\bibitem[Kass and Raftery(1995)]{kass1995bayes}
Robert~E Kass and Adrian~E Raftery.
\newblock Bayes factors.
\newblock \emph{Journal of the american statistical association}, 90\penalty0 (430):\penalty0 773--795, 1995.

\bibitem[Lehmann et~al.(1986)Lehmann, Romano, and Casella]{lehmann1986testing}
Erich~Leo Lehmann, Joseph~P Romano, and George Casella.
\newblock \emph{Testing statistical hypotheses}, volume~3.
\newblock Springer, 1986.

\bibitem[Lei and Wasserman(2014)]{lei2014distribution}
Jing Lei and Larry Wasserman.
\newblock Distribution-free prediction bands for non-parametric regression.
\newblock \emph{Journal of the Royal Statistical Society Series B: Statistical Methodology}, 76\penalty0 (1):\penalty0 71--96, 2014.

\bibitem[Lueckmann et~al.(2019)Lueckmann, Bassetto, Karaletsos, and Macke]{lueckmann2019likelihood}
Jan-Matthis Lueckmann, Giacomo Bassetto, Theofanis Karaletsos, and Jakob~H Macke.
\newblock Likelihood-free inference with emulator networks.
\newblock In \emph{Symposium on Advances in Approximate Bayesian Inference}, pages 32--53, 2019.

\bibitem[Lueckmann et~al.(2021)Lueckmann, Boelts, Greenberg, Goncalves, and Macke]{lueckmann2021benchmarking}
Jan-Matthis Lueckmann, Jan Boelts, David Greenberg, Pedro Goncalves, and Jakob Macke.
\newblock Benchmarking simulation-based inference.
\newblock In \emph{International conference on artificial intelligence and statistics}, pages 343--351. PMLR, 2021.

\bibitem[Marsaglia et~al.(2003)Marsaglia, Tsang, and Wang]{marsaglia2003evaluating}
George Marsaglia, Wai~Wan Tsang, and Jingbo Wang.
\newblock Evaluating kolmogorov's distribution.
\newblock \emph{Journal of statistical software}, 8:\penalty0 1--4, 2003.

\bibitem[Masserano et~al.(2023)Masserano, Dorigo, Izbicki, Kuusela, and Lee]{masserano2023simulator}
Luca Masserano, Tommaso Dorigo, Rafael Izbicki, Mikael Kuusela, and Ann Lee.
\newblock Simulator-based inference with waldo: Confidence regions by leveraging prediction algorithms and posterior estimators for inverse problems.
\newblock In \emph{International Conference on Artificial Intelligence and Statistics}, pages 2960--2974. PMLR, 2023.

\bibitem[McCullagh(2019)]{mccullagh2019generalized}
Peter McCullagh.
\newblock \emph{Generalized linear models}.
\newblock Routledge, 2019.

\bibitem[Meinshausen and Ridgeway(2006)]{meinshausen2006quantile}
Nicolai Meinshausen and Greg Ridgeway.
\newblock Quantile regression forests.
\newblock \emph{Journal of machine learning research}, 7\penalty0 (6), 2006.

\bibitem[Papamakarios et~al.(2019)Papamakarios, Sterratt, and Murray]{papamakarios2019likelihood}
George Papamakarios, David Sterratt, and Iain Murray.
\newblock Sequential neural likelihood: Fast likelihood-free inference with autoregressive flows.
\newblock In \emph{The 22nd {I}nternational {C}onference on {A}rtificial {I}ntelligence and {S}tatistics}, pages 837--848, 2019.

\bibitem[Patel et~al.(2024)Patel, McNamara, Loper, Regier, and Tewari]{patelvariational}
Yash Patel, Declan McNamara, Jackson Loper, Jeffrey Regier, and Ambuj Tewari.
\newblock Variational inference with coverage guarantees in simulation-based inference.
\newblock In \emph{Forty-first International Conference on Machine Learning}, 2024.

\bibitem[Pedregosa et~al.(2011)Pedregosa, Varoquaux, Gramfort, Michel, Thirion, Grisel, Blondel, Prettenhofer, Weiss, Dubourg, et~al.]{pedregosa2011scikit}
Fabian Pedregosa, Ga{\"e}l Varoquaux, Alexandre Gramfort, Vincent Michel, Bertrand Thirion, Olivier Grisel, Mathieu Blondel, Peter Prettenhofer, Ron Weiss, Vincent Dubourg, et~al.
\newblock Scikit-learn: Machine learning in python.
\newblock \emph{the Journal of machine Learning research}, 12:\penalty0 2825--2830, 2011.

\bibitem[Pereira and Stern(1999)]{pereira1999evidence}
Carlos Alberto de~Bragan{\c{c}}a Pereira and Julio~Michael Stern.
\newblock Evidence and credibility: full bayesian significance test for precise hypotheses.
\newblock \emph{Entropy}, 1\penalty0 (4):\penalty0 99--110, 1999.

\bibitem[Rezende and Mohamed(2015)]{rezende2015variational}
Danilo Rezende and Shakir Mohamed.
\newblock Variational inference with normalizing flows.
\newblock In \emph{International conference on machine learning}, pages 1530--1538. PMLR, 2015.

\bibitem[Sadinle et~al.(2019)Sadinle, Lei, and Wasserman]{sadinle2019least}
Mauricio Sadinle, Jing Lei, and Larry Wasserman.
\newblock Least ambiguous set-valued classifiers with bounded error levels.
\newblock \emph{Journal of the American Statistical Association}, 114\penalty0 (525):\penalty0 223--234, 2019.

\bibitem[Schervish(2012)]{schervish2012theory}
Mark~J Schervish.
\newblock \emph{Theory of statistics}.
\newblock Springer Science \& Business Media, 2012.

\bibitem[Shafer and Vovk(2008)]{shafer2008tutorial}
Glenn Shafer and Vladimir Vovk.
\newblock A tutorial on conformal prediction.
\newblock \emph{Journal of Machine Learning Research}, 9\penalty0 (3), 2008.

\bibitem[Stimper et~al.(2023)Stimper, Liu, Campbell, Berenz, Ryll, Schölkopf, and Hernández-Lobato]{Stimper2023}
Vincent Stimper, David Liu, Andrew Campbell, Vincent Berenz, Lukas Ryll, Bernhard Schölkopf, and José~Miguel Hernández-Lobato.
\newblock normflows: A pytorch package for normalizing flows.
\newblock \emph{Journal of Open Source Software}, 8\penalty0 (86):\penalty0 5361, 2023.
\newblock \doi{10.21105/joss.05361}.
\newblock URL \url{https://doi.org/10.21105/joss.05361}.

\bibitem[Van~der Vaart(2000)]{van2000asymptotic}
Aad~W Van~der Vaart.
\newblock \emph{Asymptotic statistics}, volume~3.
\newblock Cambridge university press, 2000.

\bibitem[Vovk et~al.(2005)Vovk, Gammerman, and Shafer]{vovk2005algorithmic}
Vladimir Vovk, Alexander Gammerman, and Glenn Shafer.
\newblock \emph{Algorithmic learning in a random world}, volume~29.
\newblock Springer, 2005.

\bibitem[Vovk et~al.(2014)Vovk, Petej, and Fedorova]{vovk2014conformal}
Vladimir Vovk, Ivan Petej, and Valentina Fedorova.
\newblock From conformal to probabilistic prediction.
\newblock In \emph{Artificial Intelligence Applications and Innovations: AIAI 2014 Workshops: CoPA, MHDW, IIVC, and MT4BD, Rhodes, Greece, September 19-21, 2014. Proceedings 10}, pages 221--230. Springer, 2014.

\bibitem[Vovk et~al.(2016)Vovk, Fedorova, Nouretdinov, and Gammerman]{vovk2016criteria}
Vladimir Vovk, Valentina Fedorova, Ilia Nouretdinov, and Alexander Gammerman.
\newblock Criteria of efficiency for conformal prediction.
\newblock In \emph{Conformal and Probabilistic Prediction with Applications: 5th International Symposium, COPA 2016, Madrid, Spain, April 20-22, 2016, Proceedings 5}, pages 23--39. Springer, 2016.

\bibitem[Vovk et~al.(2019)Vovk, Shen, Manokhin, and Xie]{Vovk2019}
Vladimir Vovk, Jieli Shen, Valery Manokhin, and Min{-}ge Xie.
\newblock Nonparametric predictive distributions based on conformal prediction.
\newblock \emph{Mach. Learn.}, 108\penalty0 (3):\penalty0 445--474, 2019.
\newblock \doi{10.1007/S10994-018-5755-8}.

\bibitem[Vovk et~al.(2020)Vovk, Petej, Nouretdinov, Manokhin, and Gammerman]{Vovk2020}
Vladimir Vovk, Ivan Petej, Ilia Nouretdinov, Valery Manokhin, and Alexander Gammerman.
\newblock Computationally efficient versions of conformal predictive distributions.
\newblock \emph{Neurocomputing}, 397:\penalty0 292--308, 2020.
\newblock \doi{10.1016/J.NEUCOM.2019.10.110}.

\bibitem[Vovk et~al.(2022)Vovk, Gammerman, and Shafer]{vovk2022algorithmic}
Vladimir Vovk, Alexander Gammerman, and Glenn Shafer.
\newblock \emph{Algorithmic Learning in a Random World}.
\newblock Springer Nature, 2022.

\bibitem[Wichitchan et~al.(2019)Wichitchan, Yao, and Yang]{wichitchan2019hypothesis}
Supawadee Wichitchan, Weixin Yao, and Guangren Yang.
\newblock Hypothesis testing for finite mixture models.
\newblock \emph{Computational statistics \& data analysis}, 132:\penalty0 180--189, 2019.

\end{thebibliography}

\appendix
\section{Experiment details}
\label{appendix:experiment_details}
In this section, we provide supplementary details regarding our experimental methodology and present additional results. We specify the statistics and models/simulators used for both the tractable and intractable likelihood comparisons and describe the architectures for the LFI posterior estimator and the estimator used in the Poisson Counting experiment.

For implementation, all experiments involving normalizing flows for posterior estimation were conducted using the \textit{normflows} Python package \citep{Stimper2023}. We employed an autoregressive rational quadratic spline architecture for each flow \citep{durkan2019neural}, and the models were trained on an Acer Nitro 5 AN515-54 using a GPU (CUDA).

\subsection{Statistics and models details}
\label{appendix:stats_details}

\subsubsection{Tractable likelihood}
\label{appendix:tractable_ll}

The data-generating processes are:

\begin{itemize}
    \item \textbf{Normal model with fixed variance}: $X_i \overset{iid}{\sim} N(\theta, 1)$, with $\theta \in \Theta = [-5, 5]$. When a prior is needed to build $\tau$, we use $\theta \sim N(0, 0.25)$.
    \item \textbf{Gaussian mixture model (GMM)}: $X_i \overset{iid}{\sim} 0.5N(\theta, 1) + 0.5N(-\theta, 1)$, with $\theta \in \Theta = [0,5]$. When a prior is needed to build $\tau$, we use $\theta \sim N(0.25, 1)$ as the prior.
    \item \textbf{Lognormal model with both mean and scale as parameters}: $X_i \overset{iid}{\sim} \text{lognormal}(\mu, \sigma^2)$, with $\boldsymbol{\theta} = \mu \times \sigma^2 \in [-2.5, 2.5] \times [0.15, 1.25] = \boldsymbol{\Theta}$. When a prior is needed to build $\tau$, we use $(\mu, \sigma^2) \sim \text{NIG}(0, 2, 2, 1)$.
\end{itemize}
To choose $\tau$, we consider:
\begin{itemize}
    \item \textbf{Likelihood ratio}: 
    Considering that $\mathcal{L}(.; \theta_0)$ represents the likelihood function, the likelihood ratio statistic is given by:
    \begin{align}
    \label{eq:likelihood_ratio}
    LR(\x, \theta_0) = \log{\frac{\mathcal{L}(\x; \theta_0)}{\sup_{\theta \neq \theta_0} \mathcal{L}(\x;\theta)}} \;.
    \end{align}
    Under regularity conditions,
    Wilks' theorem \citep{drton2009likelihood} implies that
    $-2LR(\X, \theta_0)|\theta=\theta_0$
    has an asymptotic  $\chi_1^2$ distribution, which is typically used to approximate $C_\theta$. However, regularity conditions are not always met, as in the case of the Gaussian Mixture Model \citep{chen2009hypothesis}.
    \item \textbf{Kolmogorov-Smirnov statistic}:
    Let $F_{\theta_0}(\cdot)$ be the theoretical CDF under $\theta_0$ and  $\widehat{F}_n(\cdot)$ be the empirical CDF estimated using data $\X$ with size $n$. The Kolmogorov-Smirnov statistic is defined as:
    \begin{align*}
        KS(\x, \theta_0) = \sup_{x \in \mathcal{X}} |F_{\theta_0}(x) - \widehat{F}_n(x)| = D_n \; .
    \end{align*}
    In this case, a classical result used to approximate $C_\theta$ is that, under the null hypothesis:
    \begin{align*}
        \sqrt{N} D_n \overset{n \rightarrow \infty}{\longrightarrow} \mathcal{K} \; ,
    \end{align*}
    with $\mathcal{K}$ being the Kolmogorov distribution \citep{marsaglia2003evaluating}.
    \item \textbf{Bayes Factor}:
    Considering a prior probability $\pi$ over $\Theta$ and the comparison of the hypothesis $H_0:\theta \in \Theta_0$ to its complement $H_1:\theta \in \Theta_1$, the Bayes factor (BF) is given by the ratio of the marginal likelihood of both hypothesis \citep{kass1995bayes}:
\begin{align}
\label{eq:og_bayes_factor}
    BF(\x, \Theta_0) = \frac{\P(\x|H_0)}{\P(\x|H_1)} = \frac{\int_{\Theta_0} \mathcal{L}(\x; \theta) d\pi_{0}(\theta)}{\int_{\Theta_1} \mathcal{L}(\x; \theta) d\pi_{1}(\theta)} \; ,
\end{align}
where $\pi_0$ and $\pi_1$ represent the restrictions of $\pi$ to $\Theta_0$ and $\Theta_1$, respectively. Since our focus is on the precise hypothesis $H_0: \theta = \theta_0$ without altering the joint distribution of the data $\X$ and $\theta$, we define the restrictions as follows:
$$\pi_0(\theta) = \begin{cases}
                        1 \text{ if $\theta = \theta_0$} \\
                        0 \text{ otherwise}
                    \end{cases} \quad \pi_1 = \begin{cases}
                        \pi(\theta) \text{ if $\theta \neq \theta_0$} \\
                        0 \text{ otherwise}
                    \end{cases} \, $$
and then we compute the Bayes Factor statistic as:
\begin{align}
\label{eq:bayes_frequentist_factor}
    BF(\x, \theta_0) = \frac{\int_{\Theta_0} \mathcal{L}(\x; \theta) d\pi_{0}(\theta)}{\int_{\Theta_1} \mathcal{L}(\x; \theta) d\pi_{1}(\theta)} = \frac{\mathcal{L}(\x;\theta_0)}{f(\x)} = \frac{f(\theta_0, \x)}{\pi(\theta_0) \cdot f(\x)} = \frac{f(\theta_0|\x)}{\pi(\theta_0)} \; .
\end{align}
Following \cite{dalmasso2021likelihood}, we use the Bayes Factor as a frequentist statistic to construct confidence sets. We adopt the term  ``Bayes Frequentist Statistic'', as introduced by \cite{dalmasso2021likelihood}.
    \item \textbf{E-value:} Another statistic used to test hypotheses in a Bayesian context is the Full Bayesian Significance Testing \citep{pereira1999evidence, pereira2008can}. This procedure assigns an evidence measure based on the posterior distribution to the hypothesis \( H_0 \) and rejects it if the evidence measure is small. Specifically, let \( T_{\theta_0} = \{\theta \in \Theta \mid f(\theta|\X) \geq f(\theta_0|\X)\} \). The Bayesian evidence measure (e-value) in favor of \( H_0 \) is defined as:
\begin{align}
\label{eq:e_value}
    ev(\x, \theta_0) = 1 - \P(\theta \in T_{\theta_0}|\x) = 1 - \int_{T_{\theta_0}} f(\theta|\x) d\theta.
\end{align}
Although the e-value is originally considered a Bayesian measure, we use it as a frequentist statistic to construct confidence sets. Therefore, we use the term "Frequentist e-value" to refer to this statistic in our context. \cite{diniz2012relationship} establishes an asymptotic connection between the e-value and the likelihood ratio statistic, providing an asymptotic approximation for the e-value under contour restrictions. The resulting asymptotic cutoff corresponds to the significance level \( \alpha \) for the simulation settings used here. 

\end{itemize}

\subsubsection{Intractable likelihood}
\label{appendix:intractable_ll}

We consider the following simulators:

\begin{itemize}
\item \textbf{SLCP (Simple likelihood complex posterior)}: in this model, we consider that $\x = (\x_1, \x_2, \x_3, \x_4) \in \mathbb{R}^8$ represents 2d-coordinates of 4 points. The coordinates of each point are sampled independently from a multivariate Gaussian whose mean and covariance matrix are parametrized by a 5-dimensional parameter $\boldsymbol{\theta}$ \citep{hermans2021trust}. Despite the likelihood's simplicity, the posterior of this model is complex \citep{papamakarios2019likelihood}, making it difficult to use the BFF and E-value statistics in the analytical formula. We let $\Theta = (-3,3)^5$ and use $\theta_i \sim U(-3, 3)$ as a prior for the model and to estimate $\tau$.
\item \textbf{Two Moons}:  here we consider a likelihood model with $\x = (x_1, x_2) \in \mathbb{R}^2$ and a two-dimensional $\boldsymbol{\theta}$ such that the posterior exhibits a bimodal moon shape-like structure \citep{greenberg2019automatic}. In this case, we must model and approximate the posterior to estimate the BFF and E-value statistics. We set $\Theta = (-1,1)^2$ and use $\theta_i \sim U(-1,1)$ as a prior for the model and to estimate $\tau$. 
\item \textbf{M/G/1}:  this model describes a processing and arrival queue system, where a 3-dimensional parameter $\theta$ influences both the service time per customer and the intervals between arrivals \citep{hermans2021trust}. Here, the sample $\x = (x_1,\dots, x_5) \in \mathbb{R}^5$  consists of 5 equally spaced quantiles of inter-departure times. In this case, the likelihood $\mathcal{L}(\x; \theta)$ is intractable, and computing the BFF and E-value statistics through likelihood-based inference would require high-dimensional integration \citep{blum2010non}. We set $\Theta = (0,10)^2\times(0,1/3)$ and use $(\theta_1, \theta_2, \theta_3) \sim U(0, 10)^2 \times U(0, 1/3)$ as prior for the model and to estimate $\tau$.
\item \textbf{Weinberg}: this simulator consists of a high-energy particle collision physics model \citep{hermans2021trust}. Here, our sample $\x \in \mathbb{R}$ is a measure of the Weinberg angle and we are interested in inferring the Fermi's constant $\theta \in \mathbb{R}$. As the likelihood is intractable, we must estimate all kinds of statistics for this problem. We set  $\Theta = (0.5,1.5)$, and use the prior distribution $\theta \sim U(0.5, 1.5)$ for estimating $\tau$.
\item \textbf{SIR}: In this example, we examine an epidemiological model that tracks the dynamics of individuals across three states: susceptible (S), infectious (I), and recovered or deceased (R) \citep{lueckmann2021benchmarking}.  The sample $\x = (x_1, x_2, x_3) \in \mathbb{R}^3$ represents the count of individuals in each state within a population of 1000, observed after 10 iterations of the model. The parameter $\boldsymbol{\theta}$, a 2-dimensional vector, denotes the contact rate and the mean recovery rate of the model. This case also involves an intractable likelihood, requiring estimation of all relevant statistics. We set $\Theta = (0, 0.5)^2$ and apply the prior $\theta_i \sim U(0, 0.5)$ for estimating $\tau$.
\end{itemize}
For choices of $\tau$, we consider both the E-value and BFF introduced in Appendix \ref{appendix:tractable_ll} along with the Waldo statistic \citep{masserano2023simulator}:
\begin{align}
    \text{Waldo}(\x, \theta_0) = (\E[\theta|\x] - \theta_0)^{T} \V[\theta|\x]^{-1} (\E[\theta|\x] - \theta_0) \;,
\end{align}
where $\E[\theta|\x]$ and $\V[\theta|\x]$ replace the MLE estimator $\hat{\theta}$ and its variance by the conditional mean and covariance matrix of $\theta$ given $\x$. As detailed by \cite{masserano2023simulator}, under Bayes estimator assumptions, the Waldo statistic retains the same asymptotic properties as the Wald statistic. However, for smaller $n$, Waldo can leverage consistent priors on $\theta$ to produce tighter confidence sets. This makes Waldo a compelling choice in scenarios where the likelihood is intractable.

\subsection{LFI posterior estimator}
In the LFI setting, we employed a total of six flows, each consisting of two hidden layers with 128 units per layer. A dropout probability of 0.35 was applied to each layer to prevent overfitting. For optimization, we trained the model for a maximum of 1000 epochs, with early stopping triggered after 30 epochs of no improvement. The Adam optimizer was used with a learning rate of $3 \times 10^{-4}$ and a weight decay of $1 \times 10^{-5}$. To train the neural network, we simulated $B = 30,000$ samples for all sample sizes $n$.

\subsection{Poisson Counting Experiment posterior estimator}
For the nuisance parameter experiment, we used a total of four flows with the same configuration for hidden layers, hidden units, dropout probability, and optimizer as in the LFI neural posterior model. In this case, training was conducted for up to 2000 epochs, with early stopping triggered after 200 epochs without improvement. The neural network was trained using $B = 25,000$ simulated samples for $n = 1$.

\subsection{Experiment results}
\label{appendix:exp_results}
This subsection presents heatmaps comparing the conditional coverage performance of each method for each statistic separately. In each heatmap, rows correspond to the methods and columns represent specific scenarios, defined by a model and budget combination. The heatmaps are divided by each sample size. For each scenario, we highlight the top-performing methods in green, which are those with the lowest MAE or an MAE within two standard errors of the minimum.

\subsubsection{Tractable likelihood detailed results}
\label{appendix:trac_results}

Figures \ref{fig:heatmap_bff_sim}, \ref{fig:heatmap_LR_sim}, \ref{fig:heatmap_KS_sim}, and \ref{fig:heatmap_e_value_sim} show the performance of the BFF, LR, KS, and E-value statistics, respectively, for sample sizes of $n=10, 20, 50, 100$.

Overall, \ourmethodpp{} tuned and \ourmethod{} are strong competitors for tractable models. Nonetheless, certain methods are preferable in specific scenarios. For instance, the asymptotic approach performs particularly well for the LR statistic under the normal model with fixed variance (Figure \ref{fig:heatmap_LR_sim}). Similarly, boosting shows competitive performance for the E-value: while not consistently superior across all setups, its performance improves with larger sample sizes. For the E-value, our methods perform well for $n=10, 20$, but their advantage diminishes as the sample size increases.

\begin{figure}[!ht]
    \centering
    \includegraphics[width=1\linewidth]{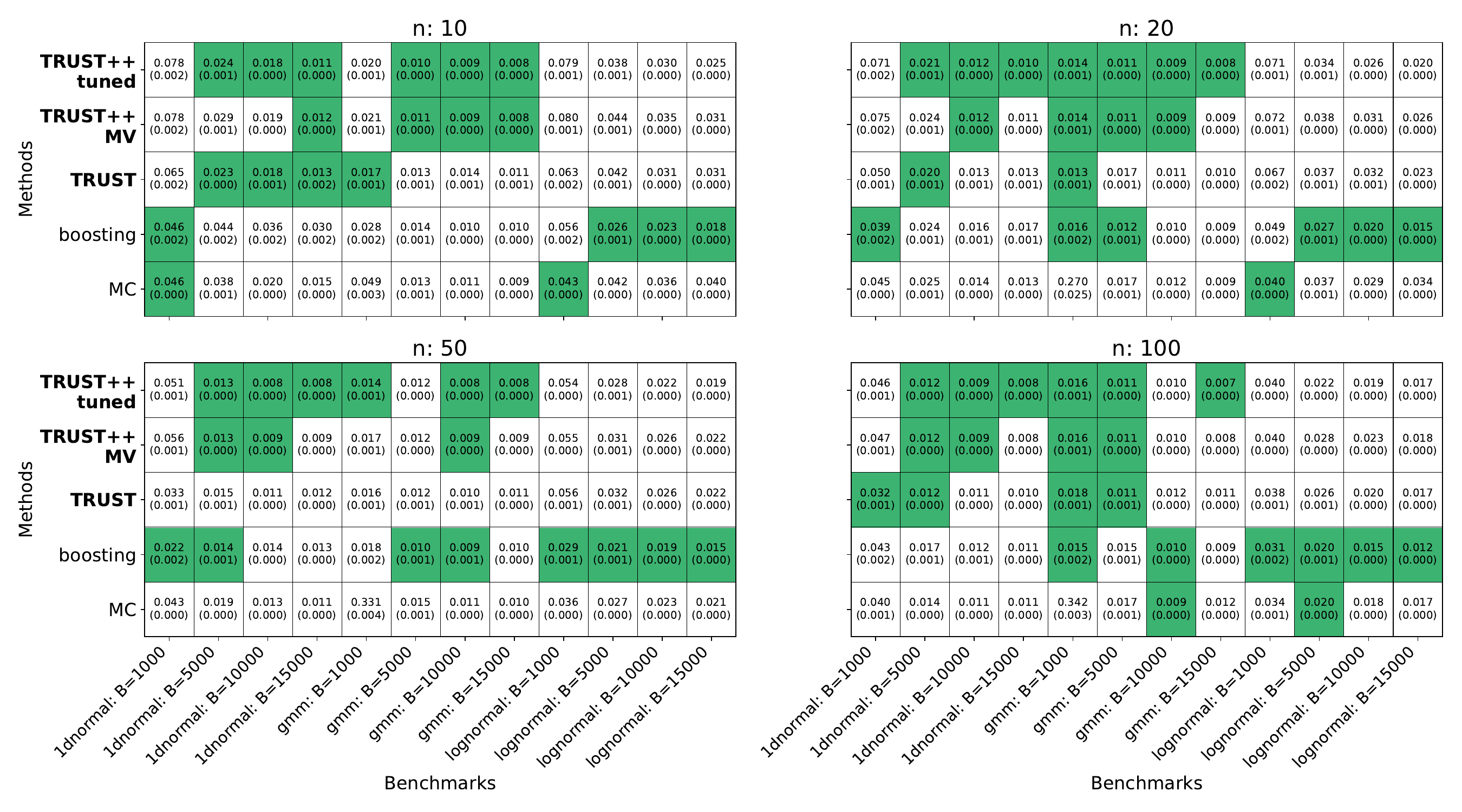}
    \caption{Best methods (lowest MAE) for BFF on tractable models. \ourmethodpp{} tuned, \ourmethodpp{}, and boosting are best performers in general. For $n=10, 100$, \ourmethod{} is also a decent competitor, and MC is among the best in a few setups, proving to be the least robust choice.}
    \label{fig:heatmap_bff_sim}
\end{figure}

\begin{figure}[!ht]
    \centering
    \includegraphics[width=1\linewidth]{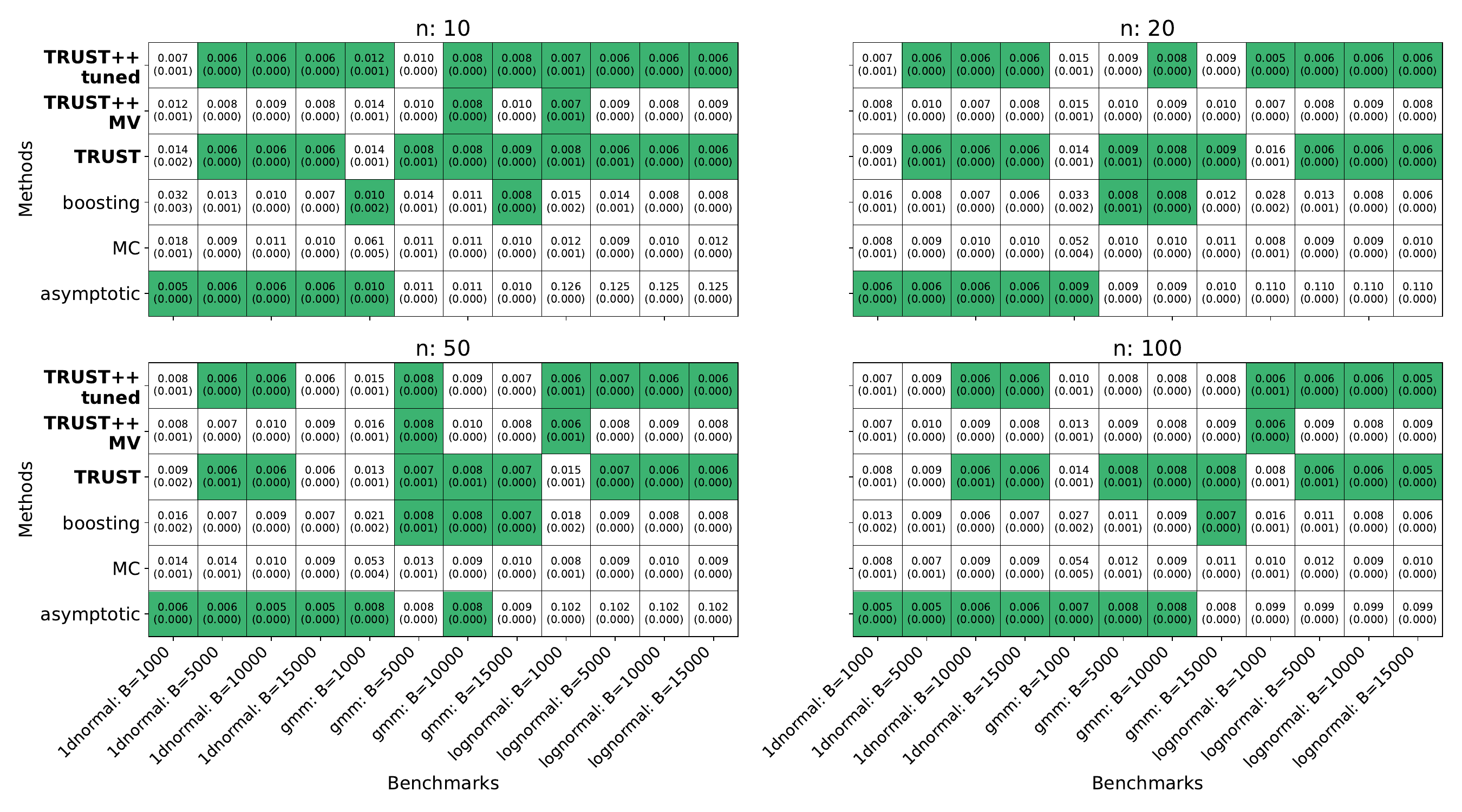}
    \caption{Best methods (lowest MAE) for LR on tractable models. The asymptotic is always among the best competitors for the Normal model with fixed variance D.G.P., and also appears in some setups for the GMM. Again, our methods are consistently good performers, especially \ourmethodpp{} tuned and \ourmethod{}. }
    \label{fig:heatmap_LR_sim}
\end{figure}

\begin{figure}[!ht]
    \centering
    \includegraphics[width=1\linewidth]{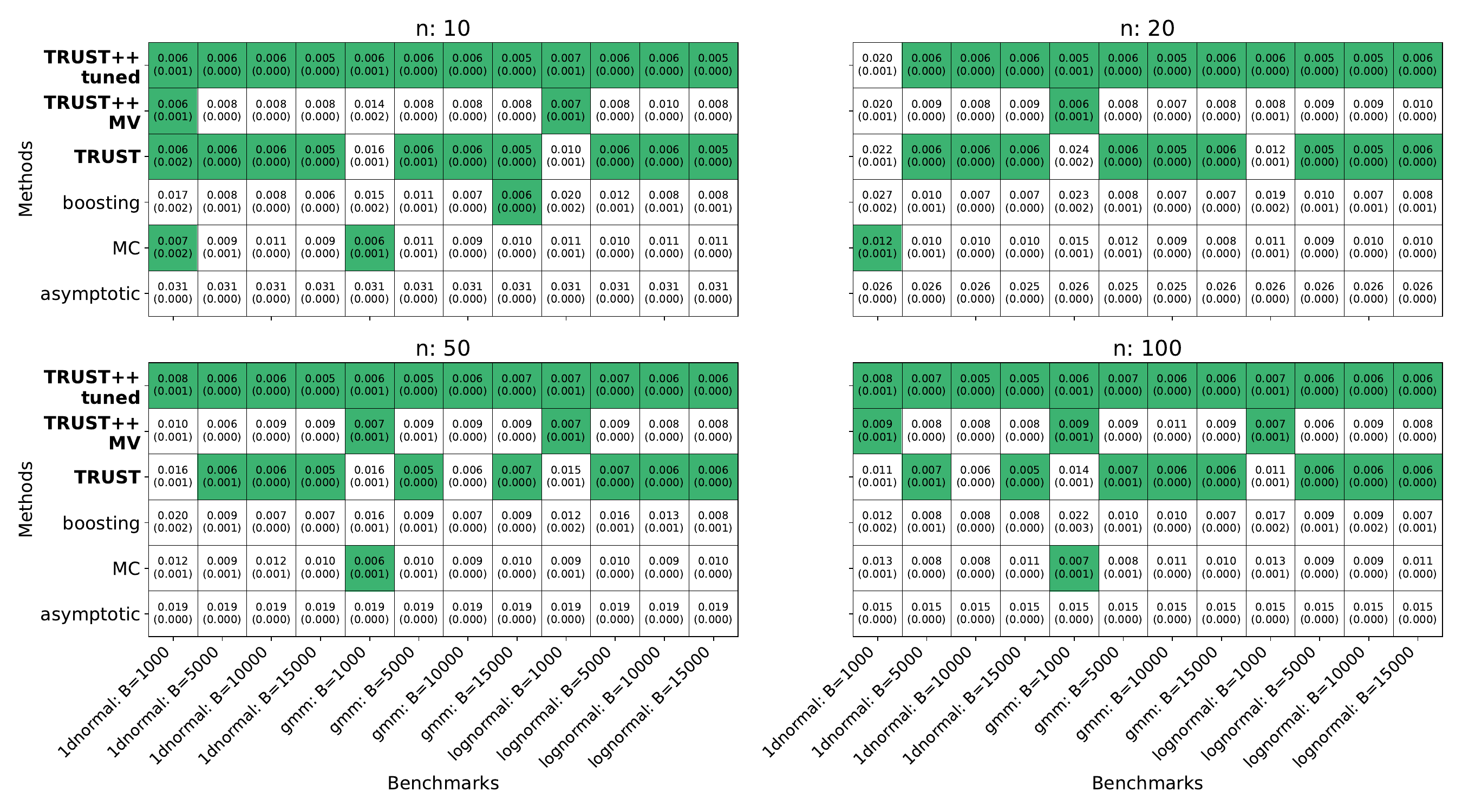}
    \caption{Best methods (lowest MAE) for KS on tractable models. \ourmethodpp{} tuned is among the best methods in 47 out of 48 setups, proving to be a powerful competitor. \ourmethod{} is also very consistent, appearing in the majority os scenarios. MC and \ourmethodpp{} MV have appeared less than the other two, but are considerably more present than boosting and asymptotic, which are not a good choice for this statistic, as shown in the graph.}
    \label{fig:heatmap_KS_sim}
\end{figure}

\begin{figure}[!ht]
    \centering
    \includegraphics[width=1\linewidth]{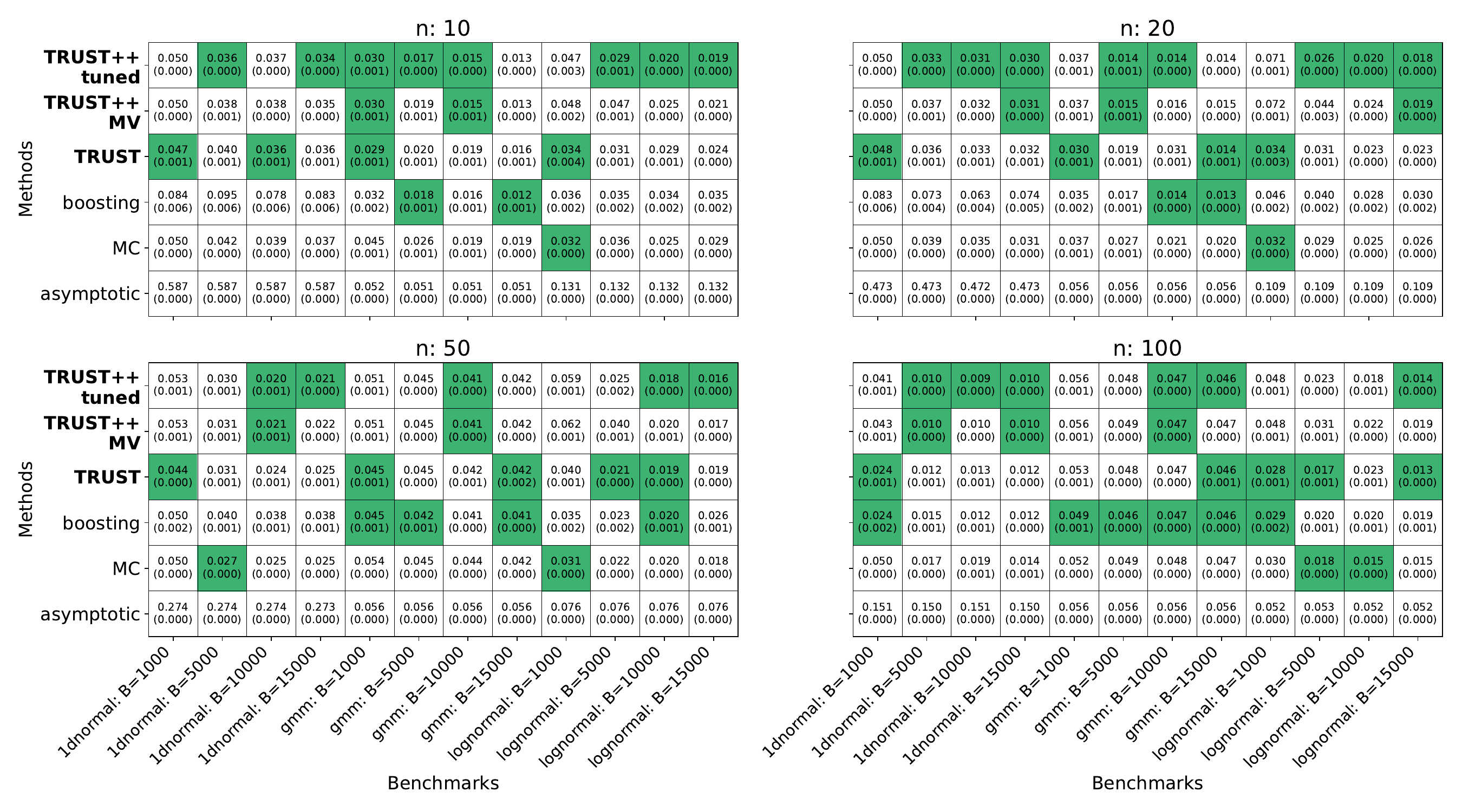}
    \caption{Best methods (lowest MAE) for E-value on tractable models. In this case, we do not see a specific method that shows a clear dominance as in Figure \ref{fig:heatmap_KS_sim}. However, it is possible to notice the poor performance of asymptotic, which is never among the best methods. Our  methods and the boosting are consistent across the setups.}
    \label{fig:heatmap_e_value_sim}
\end{figure}

\subsubsection{Intractable likelihood detailed results}
\label{appendix:intrac_results}

Figures \ref{fig:heatmap_bff_real}, \ref{fig:heatmap_e_value_real}, and \ref{fig:heatmap_waldo_real} present the models’ performance for the BFF, E-value, and Waldo statistics, respectively, in the LFI experiments with sample sizes $n=1, 5, 10, 20$.

\ourmethodpp{} tuned, \ourmethodpp{} MV, and boosting emerge as the overall best performers. \ourmethod{} achieves reasonable results in certain scenarios (e.g., Waldo on the SIR benchmark with $n=10$) but rarely ranks among the top methods. MC, in contrast, consistently performs poorly for intractable likelihoods.

For $n=1$, \ourmethodpp{} tuned, \ourmethodpp{} MV, and boosting achieve strong results across all statistics. Interestingly, \ourmethod{} performs well for the E-value in the tractable generator, showing good coverage control for $B = 15{,}000, 20{,}000, 30{,}000$.

\begin{figure}[!ht]
    \centering
    \includegraphics[width=1\linewidth]{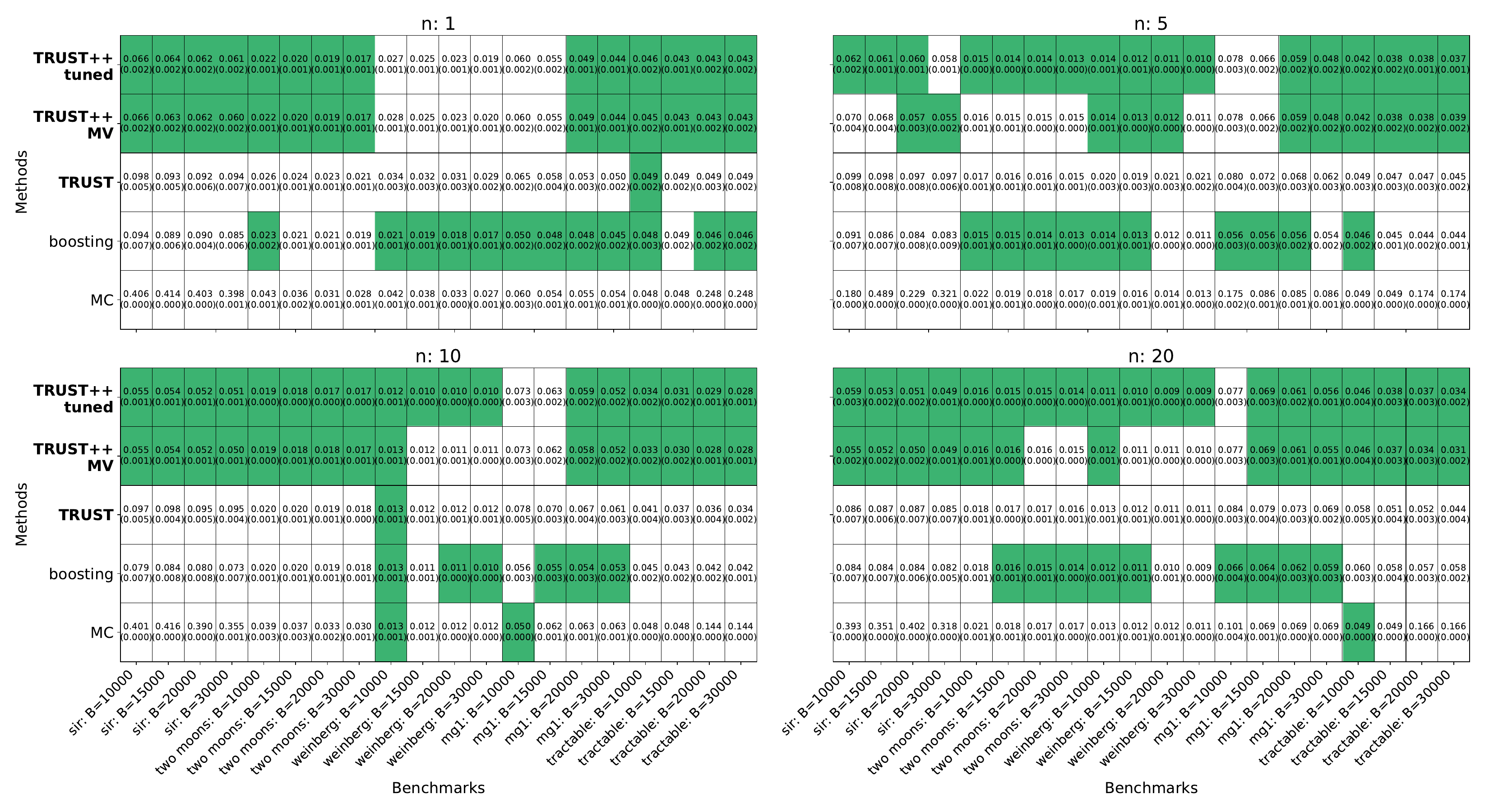}
    \caption{Best methods (lowest MAE) for BFF on intractable models. \ourmethodpp{} tuned, \ourmethodpp{} MV and boosting are the strongest competitors in this scenario. In 79 out of 80 experiments, at least one of our methods is highlighted as the best performer. MC and \ourmethod{} performances are consistently poorer than their competitors.}
    \label{fig:heatmap_bff_real}
\end{figure}

\begin{figure}[!ht]
    \centering
    \includegraphics[width=1\linewidth]{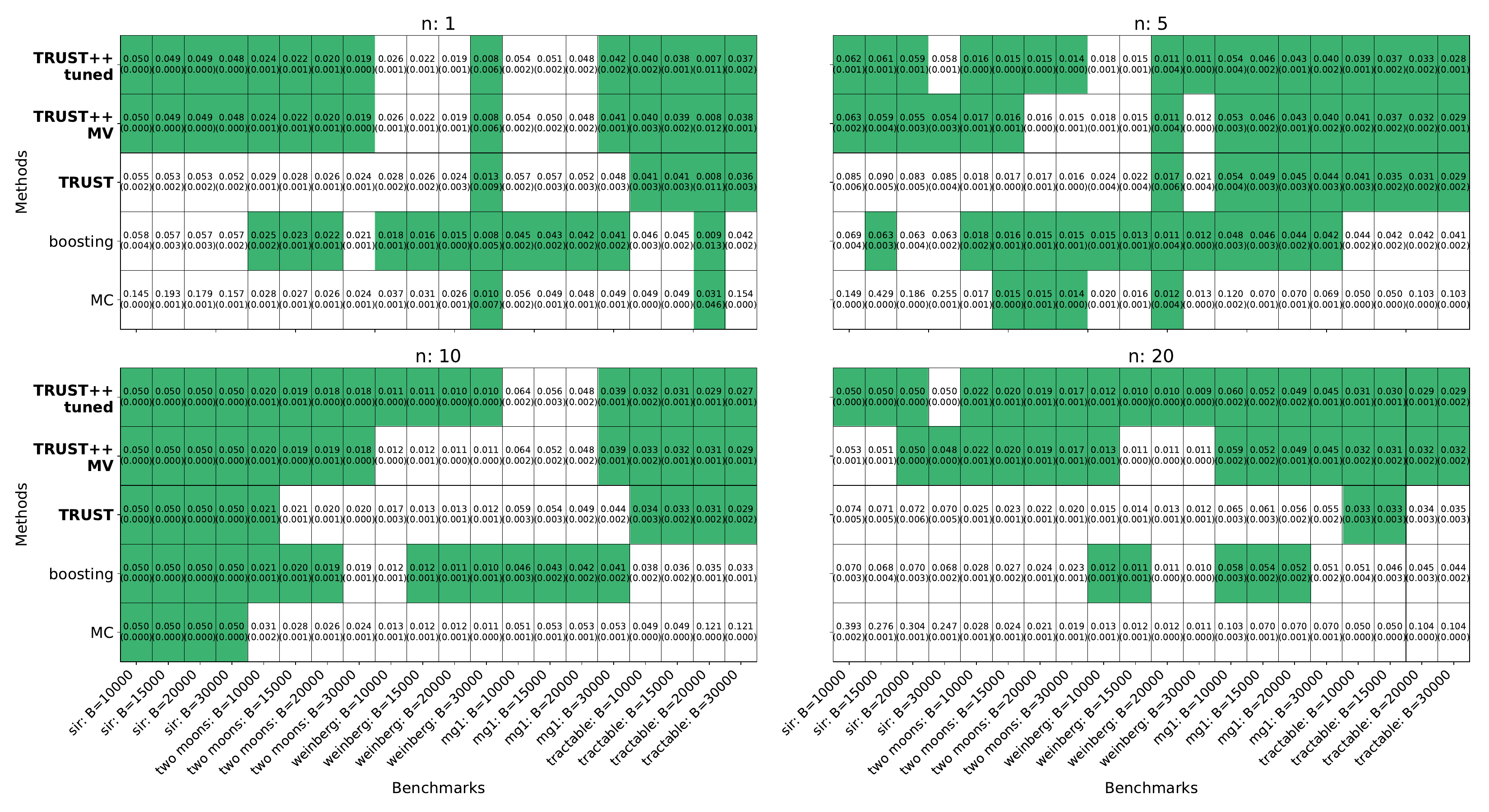}
    \caption{Best methods (lowest MAE) for E-value on intractable models. As on BFF, \ourmethodpp{} tuned, \ourmethodpp{} MV and boosting are the most consistent competitors. However, for this statistic, \ourmethod{} stands out in some scenarios, especially when $n=5, 10$ (i.e., not in the extreme cases).}
    \label{fig:heatmap_e_value_real}
\end{figure}

\begin{figure}[!ht]
    \centering
    \includegraphics[width=1\linewidth]{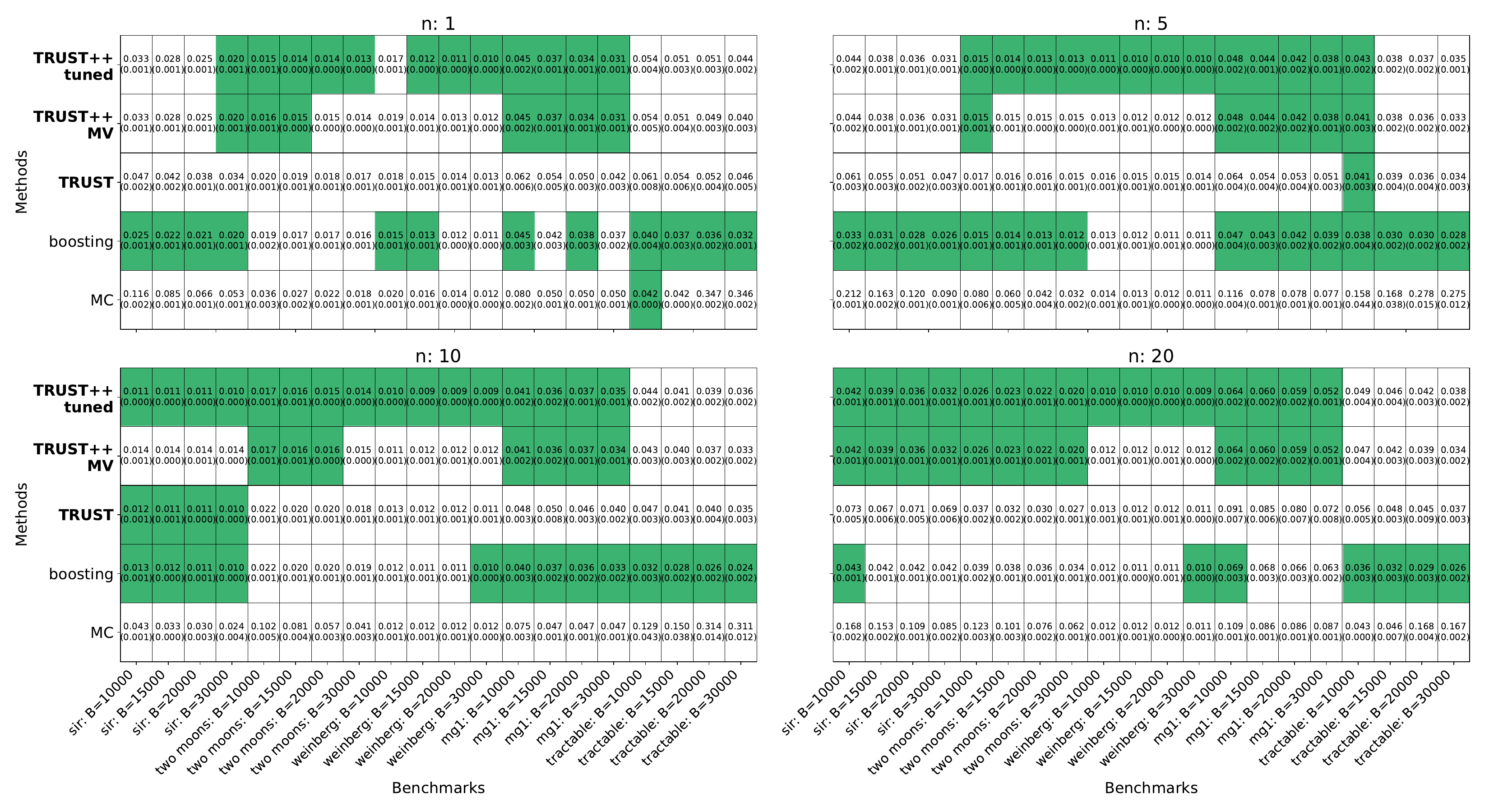}
    \caption{Best methods (lowest MAE) for Waldo on intractable models. Similarly to BFF, \ourmethodpp{} tuned, \ourmethodpp{} MV and boosting are the methods that stand out the most. MC only appears as the winner one time out of 80 experiments, and \ourmethod{} has a modest performance, with good results for the SIR benchmark with $n=10$.}
    \label{fig:heatmap_waldo_real}
\end{figure}

\section{Algorithm for tuning $M$}
\label{appendix:alg_tune_M}
\begin{algorithm}[hbt!]
\caption{Tuning algorithm for $M$}\label{alg:tune_M}
\KwData{Validation grid size $B_{\text{tune}}$; grid of $M$'s between $0$ to $K$ $M_{\text{grid}}$; fitted \ourmethodpp{}; number of simulated statistics $n_{\text{sim}}$ in each grid point}
\KwResult{Tuned value of $M$, $M_{\text{tuned}}$}
simulate $\Theta_{\text{grid}} = (\theta_1, \dots, \theta_{B_{\text{tune}}})$ from the prior \;
\For{$M_{\text{cand}} \in M_{\text{grid}}$}{
compute $\widehat{C}_{\theta}$ in \ourmethodpp{} with $M$ fixed as $M_{\text{cand}}$ for each $\theta \in \Theta_{\text{grid}}$ \;
compute $\text{cover}_{\alpha}(\widehat{C}, \theta)$ for each $\theta \in \Theta_{\text{grid}}$ (as in Eq. \ref{eq:cover_cutoff})\;
$MAE_{M_{\text{cand}}} \gets MAE(\widehat{C}, \alpha)$ (as in Eq. \ref{eq:MAE_cutoff}) \;
}
$M_{\text{tuned}} \gets \mathop{\arg \min}\limits_{M_{\text{cand}} \in M_{\text{grid}}} \left\{ MAE_{M_{\text{cand}}} \right\}$ \;
\Return{$M_{\text{tuned}}$}
\end{algorithm}

\section{Proofs}
\label{appendix:proofs}

\subsection{Section \ref{sec:methodology} - Methodology}

\begin{proof}[Proof of Theorem \ref{thm:uncertainty_coverage}]
    Notice that 
    \begin{align*}
        \P\left(\theta \in \mathcal{I}(\x)|\theta \notin R(\x) \right)
    &=\P\left(\tau(\x,\theta) \geq  \widehat{C}^U_\theta| \tau(\x,\theta) \leq  C_\theta \right) \\
    &=\P\left(\widehat{C}^U_\theta \leq C_\theta | \tau(\x,\theta) \leq  C_\theta \right) \leq \beta/2,
    \end{align*}
    which proves the first statement. The proof for the second statement is analogous.
\end{proof}

\subsection{Section \ref{sec:theory} - Theoretical Results}


To summarize, in \ourmethod{}, the partition is created by building a regression tree that uses \( \theta \) as the input and \( \tau(\x, \theta) \) as the output. This tree, trained on the simulated data 
\[(\theta_1, \tau(\X_1, \theta_1)), \ldots, (\theta_B, \tau(\X_B, \theta_B)),\] 
naturally induces a partition on \( \Theta \).

For \ourmethodpp{}, we extend this approach by generating \( K \) regression trees. We then use Breiman's proximity measure \(\rho(\theta', \theta)\), which counts the number of times \(\theta'\) and \(\theta\) appear together in the same leaf across these \( K \) trees. Finally, we define the partition \(\mathcal{A}\) based on this proximity measure, using the equivalence relation \(\theta \sim \theta' \iff \rho(\theta', \theta) = K\); in other words, \(\theta\) and \(\theta'\) must appear together in the same leaves across all \( K \) trees.

This framework ensures that the partition \(\mathcal{A}\) meets the necessary structural conditions to apply conformal prediction as stated in Theorem \ref{thm:partition-based-coverage}.

\begin{proof}[Proof of Theorem \ref{thm:partition-based-coverage}]
This is a straightforward application of conformal prediction (see \cite{angelopoulos2023conformal}). The complete proof can be found in detail in \cite[Theorem 2]{cabezas2025regression}.
\end{proof}

To prove Theorem \ref{thm:asympt_conditional} we need the following lemma. 

\begin{Lemma}\label{lemma:expectation_to_conditional_prob}
Let \( H(t|\theta) = \P(\tau(\X, \theta) \leq t | \theta) \) and \( A \) be a measurable set. Suppose \( p(x, \theta) \) is the joint PDF of \((\X, \theta)\), satisfying \(\int p(x, \theta) \, dx = r(\theta) > 0\). Then we have
\[
\E[H(t|\theta) | \theta \in A] = \P(\tau(\X, \theta) \leq t | \theta \in A).
\]
\end{Lemma}
\begin{proof}
Define \(c \coloneq \P(\theta \in A)\), the result follows from a straightforward calculation, as shown next.
\begin{align*}
     \E[H(t|\theta)|\theta \in A] &= \frac{1}{\P(\theta \in A)}\E[H(t|\theta) \mathbb{I}[\theta \in A]]\\
     &= c^{-1}\int H(t|\theta) \mathbb{I}[\theta \in A]r(\theta)d\theta\\
     &= c^{-1}\int \P(\tau(\X,\theta)\leq t | \theta) \mathbb{I}[\theta \in A]r(\theta)d\theta\\
     &= c^{-1}\int \left(\frac{1}{r(\theta)}\int \mathbb{I}[\tau(x,\theta)\leq t] p(x,\theta) dx \right) \mathbb{I}[\theta \in A]r(\theta)d\theta\\
     &= c^{-1}\int \int \mathbb{I}[\tau(x,\theta)\leq t] \mathbb{I}[\theta \in A] p(x,\theta)dx d\theta\\
     &= c^{-1} \P(\tau(\X,\theta)\leq t, \theta \in A)\\
     &= \P(\tau(\X,\theta)\leq t | \theta \in A).
\end{align*}

\end{proof}

Under the Assumption \ref{assumption:strong_consistency}, for a sufficiently large \( B \), we achieve \( H(t | \theta') \approx \widehat{H}_B(t | \theta') \), for any $\theta'\in\Theta$ and $t\in\R$.  Based on this approximation, we outline here the intuition for how Theorem \ref{thm:asympt_conditional} will be established in detail later. 

Specifically, for a fixed \( \theta \) of interest, by Lemma \ref{lemma:expectation_to_conditional_prob}:
\begin{align*}
    \P(\tau(\X, \theta') \leq t \mid \theta' \in A(\theta))
    &= \E[H(t | \theta') \mid \theta' \in A(\theta)] \\
    &\approx \E[\widehat{H}_B(t | \theta') \mid \theta' \in A(\theta)],
\end{align*}
where the first equality results from a straightforward calculation, with its proof provided in the appendix.

By the construction of the partitions in \ourmethod{} and \ourmethodpp{}, if \( \theta' \in A(\theta) \), then \( \widehat{H}_B(t | \theta') = \widehat{H}_B(t | \theta) \). Therefore,
\begin{align*}
    \P(\tau(\X, \theta') \leq t \,|\, \theta' \in A(\theta)) &\approx \E[\widehat{H}_B(t | \theta') \,|\, \theta' \in A(\theta)] \\
    &= \E[\widehat{H}_B(t | \theta) \,|\, \theta' \in A(\theta)] \\
    &\approx \E[H(t | \theta) \,|\, \theta' \in A(\theta)] \\
    &= H(t | \theta).
\end{align*}
This  implies that, as long as our approximation \( \widehat{H}_B\) closely matches \( H \), our partition-based estimate will serve as a reliable approximation of the test statistics distribution. In particular, as discussed in Section \ref{sec:partitioned_based_estimate}, let \( \widehat{C}_{\theta,B} = \widehat{H}^{-1}_B(\alpha | \theta) \) denote the adjusted \( \alpha \)-quantile of the values \( \{\tau(\X_b, \theta_b) : b \in I_{A(\theta)}\} \). Then,
\begin{align*}
 \P\left( \theta \not\in \widehat R_B(\X)|\theta\right) = H(\widehat{C}_{\theta,B} | \theta) \approx \P(\tau(\X, \theta') \leq \widehat{C}_{\theta,B} \,|\, \theta' \in A(\theta)) = \alpha,
\end{align*}
suggesting that our methods indeed should achieve  optimal coverage.

\begin{proof}[Proof of Theorem \ref{thm:asympt_conditional}]

    Fix $\theta$. By Assumption \ref{assumption:strong_consistency}, for any $\varepsilon,\delta>0$, there exists $B_0$ and a subset $\Gamma \subset (\Theta\times \mathcal{X})^B$ such that $\P(\Gamma)\geq 1-\delta$ and, conditionally on $\Gamma$,
\[\sup_{\tilde{t}\in\R, \tilde{\theta}\in\Theta}\left|\widehat{H}_B(\tilde{t}|\tilde{\theta})-H(\tilde{t}|\tilde{\theta})\right|\leq\varepsilon.\]
By Lemma \ref{lemma:expectation_to_conditional_prob},
\begin{align*}
    \P(\tau(\X,\theta')\leq t | \theta' \in A(\theta)) &= \E[H(t|\theta')\,|\, \theta' \in A(\theta)].
\end{align*}
Conditionally on $\Gamma$, we know that 
\[H(t|\theta') \geq \widehat{H}_B(t|\theta') - \varepsilon\]
holds uniformly for any possible value of $\theta' \in \Theta, t \in \R$. Thus,
\begin{align*}
    \P(\tau(\X,\theta')\leq t | \theta' \in A(\theta),\Gamma) &= \E[H(t|\theta')\,|\, \theta' \in A(\theta),\Gamma] \\
    &\geq \E[\widehat{H}_B(t|\theta') - \varepsilon\,|\, \theta' \in A(\theta),\Gamma] \\
    &= \E[\widehat{H}_B(t|\theta')\,|\, \theta' \in A(\theta),\Gamma] - \varepsilon.
\end{align*}
Note that if $\theta' \in A(\theta)$, then \( A(\theta') = A(\theta) \), hence \( \widehat{H}_B(t|\theta') = \widehat{H}_B(t|\theta) \). Therefore,
\begin{align*}
    \P(\tau(\X,\theta')\leq t | \theta' \in A(\theta),\Gamma) &\geq \E[\widehat{H}_B(t|\theta)\,|\, \theta' \in A(\theta),\Gamma] - \varepsilon \\
    &= \E[\widehat{H}_B(t|\theta)\,|\, \theta' \in A(\theta),\Gamma] - \varepsilon.
\end{align*}
Using the fact that, conditionally on $\Gamma$, 
\[\widehat{H}_B(t|\theta) \geq H(t|\theta) - \varepsilon,\]
we obtain
\begin{align*}
    \P(\tau(\X,\theta')\leq t | \theta' \in A(\theta),\Gamma) &\geq \E[\widehat{H}_B(t|\theta)\,|\, \theta' \in A(\theta),\Gamma] - \varepsilon \\
    &\geq \E[H(t|\theta)\,|\, \theta' \in A(\theta),\Gamma] - 2\varepsilon.
\end{align*}
Since \( H(t|\theta) \) is constant, it follows that
\begin{align*}
    \P(\tau(\X,\theta')\leq t | \theta' \in A(\theta),\Gamma) &\geq \E[H(t|\theta)\,|\, \theta' \in A(\theta),\Gamma] - 2\varepsilon \\
    &= H(t|\theta)\E[1\,|\, \theta' \in A(\theta),\Gamma] - 2\varepsilon \\
    &= H(t|\theta) - 2\varepsilon.
\end{align*}
Thus, for any $t \in \R$:
\begin{align*}
    \P(\tau(\X,\theta')\leq t | \theta' \in A(\theta)) &\geq \P(\tau(\X,\theta')\leq t | \theta' \in A(\theta),\Gamma)\P(\Gamma) \\
    &\geq (1-\delta)(H(t|\theta) - 2\varepsilon) \\
    &= H(t|\theta) - \delta H(t|\theta) - 2(1-\delta)\varepsilon \\
    &\geq H(t|\theta) - \delta - 2\varepsilon + 2\delta\varepsilon \\
    &\geq H(t|\theta) - \delta - 2\varepsilon.
\end{align*}
In a similar manner, we can show that 
\begin{align*}
    \P(\tau(\X,\theta')\leq t | \theta' \in A(\theta),\Gamma) &\leq H(t|\theta) + 2\varepsilon.
\end{align*}
Following \cite[Assumption 3]{meinshausen2006quantile}, we assume there exists \(0 < \gamma < 0.5\) such that \(\P(\theta' \in A(\theta)) \geq \gamma B\). This assumption is reasonable, as it implies that the partitions (i.e., the leaves) generated by the tree-based estimators are uniformly well-populated. In other words, each leaf contains a sufficient number of observations, ensuring that the coverage properties hold consistently across the partitions. Therefore,
\begin{align*}
    \P(\tau(\X,\theta')\leq t | \theta' \in A(\theta)) &\leq \P(\tau(\X,\theta')\leq t | \theta' \in A(\theta),\Gamma)\P(\Gamma) + \frac{\delta}{\gamma B} \\
    &\leq \P(\tau(\X,\theta')\leq t | \theta' \in A(\theta),\Gamma) + \delta \\
    &\leq H(t|\theta) + 2\varepsilon + \delta.
\end{align*}
Thus, for any $\varepsilon,\delta > 0$ and any $t \in \R$:
\[
| \P(\tau(\X,\theta')\leq t | \theta' \in A(\theta)) - H(t|\theta) | \leq 2\varepsilon + \delta.
\]
Taking \( t = \widehat{C}_{\theta,B}(1-\alpha) \), by the construction of the empirical quantile over $A(\theta)$, 
\[
|(1-\alpha) - H(\widehat{C}_{\theta,B}(1-\alpha) | \theta)| \leq 2\varepsilon + \delta,
\]
which is equivalent to
\[
|(1-\alpha) - \P\left(\theta \in \widehat{R}_B(\X) | \theta \right)| \leq 2\varepsilon + \delta.
\]
Since we can make $\varepsilon$ and $\delta$ arbitrarily small by choosing $B$ large enough, we conclude that
\[
\lim_{B\to \infty} \P\left(\theta \in \widehat{R}_B(\X) | \theta \right) = 1 - \alpha.
\]

\end{proof}

\end{document}